\documentclass[10pt,journal,compsoc]{IEEEtran}

\usepackage{bigstrut,array,multirow,tabularx}
\usepackage{diagbox}
\usepackage{slashbox}
\usepackage{makecell}
\usepackage{amsthm}
\usepackage{tikz}
\usepackage[super]{nth}
\usepackage{algorithm}
\usepackage{algorithmic}
\usepackage{graphicx}
\usepackage{verbatim}
\usetikzlibrary{patterns}

\usepackage{amsmath}
 
\usepackage{amssymb}
\usepackage{amsthm}
\usepackage[normalem]{ulem}
\usepackage[hidelinks]{hyperref}
\usepackage[nameinlink, capitalise, noabbrev]{cleveref}
\usepackage{mathtools}
\usepackage{dsfont}

\DeclarePairedDelimiter\ceil{\lceil}{\rceil}

\usepackage{tikz}
\usepackage{pgfplots} 
\pgfplotsset{compat=1.16} 
\usepackage{bigstrut,array,multirow,tabularx}
\usepackage{caption}
\usepackage{slashbox}

\newtheorem{definition}{Definition}[section]
\newtheorem{theorem}{Theorem}[section]
\newtheorem{remark}{Remark}

\newtheorem{proposition}[theorem]{Proposition}
\newcounter{example}[section]
\newenvironment{example}[1][]{\refstepcounter{example}\par\medskip
   \noindent \textbf{Example~\theexample. #1} \rmfamily}{\medskip}

\usepackage{amsmath}

\usepackage{amsfonts}

\usepackage{multirow}
\usepackage{multicol}
\usepackage{caption}
\usepackage{array}
\usepackage{booktabs}
\usepackage{graphicx}

\theoremstyle{theorem}

\definecolor{bblue}{HTML}{FFE333}
\definecolor{rred}{HTML}{7FFF00}
\definecolor{ggreen}{HTML}{87CEEB}

\ifCLASSOPTIONcompsoc
    \usepackage[caption=false, font=normalsize, labelfont=sf, textfont=sf]{subfig}
\else
\usepackage[caption=false, font=footnotesize]{subfig}
\fi
\usepackage{mwe}

\ifCLASSOPTIONcompsoc
  \usepackage[nocompress]{cite}
\else
  \usepackage{cite}
\fi
\ifCLASSINFOpdf
\else
\fi
\begin{document}
%
\title{On the Power of Gradual Network Alignment Using Dual-Perception Similarities}



\author{Jin-Duk Park, Cong Tran, Won-Yong Shin, {\em Senior Member}, {\em IEEE}, and Xin Cao, {\em Member}, {\em IEEE}
\IEEEcompsocitemizethanks{
\IEEEcompsocthanksitem Jin-Duk Park is with the School of Mathematics and Computing (Computational Science and Engineering), Yonsei University, Seoul 03722, South Korea.
 \protect\\
E-mail: jindeok6@yonsei.ac.kr
\IEEEcompsocthanksitem Cong Tran is with the Faculty of Information Technology, Posts and Telecommunications Institute of Technology, Hanoi 100000, Vietnam. \protect\\
E-mail: congtt@ptit.edu.vn
\IEEEcompsocthanksitem Won-Yong Shin is with the School of Mathematics and Computing (Computational Science and Engineering), Yonsei University, Seoul 03722, South Korea, and the Graduate School of Artificial Intelligence, Pohang University of Science and Technology (POSTECH), Pohang 37673, South Korea.
 \protect\\
 E-mail: wy.shin@yonsei.ac.kr
\IEEEcompsocthanksitem Xin Cao is with the School of Computer Science and Engineering, The University of New South Wales, Sydney 2052, Australia. \protect\\
E-mail: xin.cao@unsw.edu.au

(Corresponding author: Won-Yong Shin.)}}


%
%

\markboth{}%
{Shell \MakeLowercase{\textit{et al.}}: Bare Demo of IEEEtran.cls for Computer Society Journals}
%



%

\IEEEtitleabstractindextext{%
\begin{abstract}

Network alignment (NA) is the task of finding the correspondence of nodes between two networks based on the network structure and node attributes. Our study is motivated by the fact that, since most of existing NA methods have attempted to discover all node pairs {\em at once}, they do not harness information enriched through {\em interim} discovery of node correspondences to more accurately find the next correspondences during the node matching. To tackle this challenge, we propose \textsf{Grad-Align}, a new NA method that {\em gradually} discovers node pairs by making full use of node pairs exhibiting strong consistency, which are easy to be discovered in the early stage of gradual matching. Specifically, \textsf{Grad-Align} first generates node embeddings of the two networks based on graph neural networks along with our {\em layer-wise reconstruction loss}, a loss built upon capturing the first-order and higher-order neighborhood structures. Then, nodes are gradually aligned by computing {\em dual-perception similarity} measures including the {\em multi-layer embedding similarity} as well as the {\em Tversky similarity}, an asymmetric set similarity using the Tversky index applicable to networks with different scales. Additionally, we incorporate an edge augmentation module into \textsf{Grad-Align} to reinforce the structural consistency. Through comprehensive experiments using real-world and synthetic datasets, we empirically demonstrate that \textsf{Grad-Align} consistently outperforms state-of-the-art NA methods.
\end{abstract}

}

\maketitle
%
\IEEEpeerreviewmaketitle

\IEEEraisesectionheading{\section{Introduction}\label{sec:introduction}}

\subsection{Background and Motivation}

\IEEEPARstart{M}{ultiple} networks are ubiquitous in various real-world application domains, ranging from computer vision, bioinformatics, web mining, chemistry to social network analyses \cite{zhang2016final, trung2020adaptive}. Considerable attention has been paid to conducting network alignment (NA) (also known as graph matching), which is the task of finding the node correspondence across different networks and is often the very first step to perform downstream machine learning (ML) tasks on multiple networks in such applications, thus leading to more precise analyses. In social networks, the identification of different accounts (e.g., Facebook, Twitter, and Foursquare) of the same user facilitates friend recommendation, user behavior prediction, and personalized advertisement \cite{zhang2016final, trung2020adaptive,zhou2018deeplink}. For example, by discovering the correspondence between Twitter and Foursquare networks of the same user, we can improve the performance of friend/location recommendations for Foursquare users whose social connections and activities can be very sparse \cite{kong2013inferring}. As another example, in bioinformatics, aligning tissue-specific protein-protein interaction (PPI) networks can be effective in solving the problem of candidate gene prioritization \cite{ni2014inside}.

Despite the effectiveness and utilities of NA, performing the NA task poses several practical challenges. First, a fundamental assumption behind existing NA methods is the structural and/or attribute {\em consistencies}. That is, the same node is assumed to have a consistency over its connectivity structure and/or its metadata (i.e., attributes) across different networks \cite{zhang2016final}. However, such consistency constraints are not often satisfied in real-world applications. Examples of structural consistency violations include but are not limited to the following cases: 1) a user might have more connections in one social network site (e.g., Facebook) than those in another site (e.g., LinkedIn) and 2) the same gene might exhibit considerably different interaction patterns across different tissue-specific PPI networks \cite{zhang2016final, ni2014inside}. Moreover, there exist users who deliberately use different usernames across multiple social networks \cite{liu2016aligning}, which violates the assumption of attribute consistency. Thus, not all ground truth cross-network node pairs are always {\em strongly consistent}. Our study is basically initiated by the fact that most of existing NA methods (e.g., \cite{zhang2016final,zhou2018deeplink,man2016predict,heimann2018regal,singh2008global,bayati2009algorithms}) have attempted to discover all node pairs {\em at once} based on modeling their own similarity between cross-network node pairs (refer to Fig. \ref{fig1}), which thereby may not take advantage of the information enriched through {\em interim} discovery of node pairs in order to more accurately find the next node pairs during the node matching. The motivation of our study is that strongly consistent node pairs, which are easy to be found, can be very informative when discovering node pairs having weak consistency; how to exploit the information of the strongly consistent node pairs remains a technical challenge in the NA task.

Second, more importantly, source and target networks often manifest {\em different scales} in terms of the number of nodes. The NA task has been carried out using  benchmark datasets consisting of two imbalanced networks (see \cite{zhang2016final,trung2020adaptive,liu2016aligning,heimann2018regal,zhou2018deeplink}, and references therein)---for example, the Douban Online and Douban Offline networks have 3,906 and 1,118 users, respectively \cite{zhong2012comsoc}. Unfortunately, the problem raised by such networks with different scales exacerbates the inconsistency of cross-network node pairs, which thus results in a low alignment accuracy. It is another open challenge how to overcome this network imbalance problem in conducting NA.

Motivated by the above-mentioned open challenges, we push forward the state-of-the-art performance by designing a new NA method that universally shows the superior performance even including the case of networks with different scales.

\subsection{Main Contributions}
 \begin{figure}
    \centering
    \includegraphics[scale=0.14]{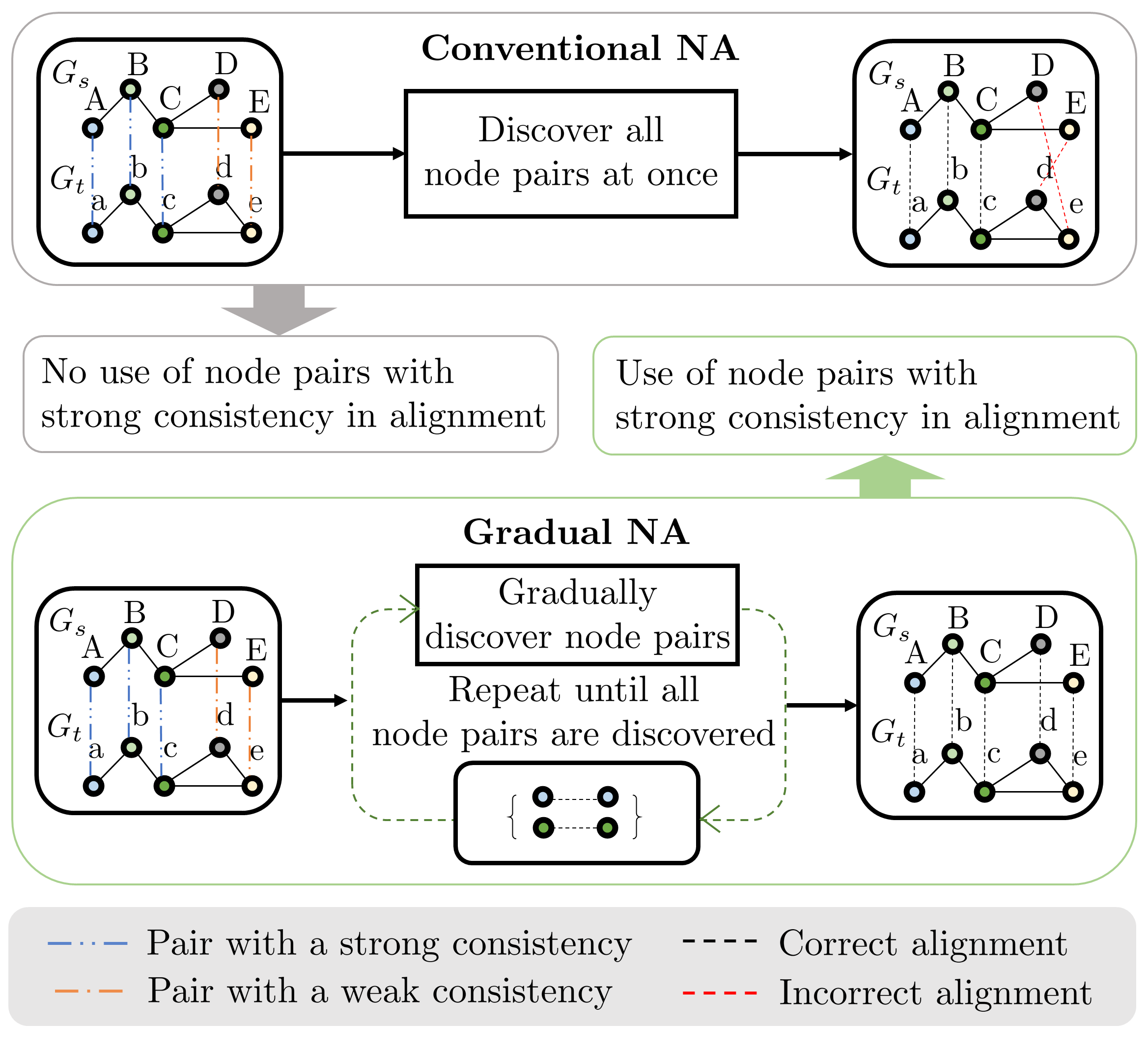}
    \caption{Comparison of conventional NA methods and our \textsf{Grad-Align} method.}
    \label{fig1}
\end{figure} 
In this paper, we introduce \textsf{Grad-Align}, a novel NA method that {\em gradually} finds the node correspondence.\footnote{The source code used in this paper is available online
(https://github.com/jindeok/Grad-Align-full).} \textsf{Grad-Align} discovers only a part of node pairs iteratively until all node pairs are found in order to combat the node inconsistency across two different networks, while fully exploiting the information of already aligned node pairs having strong consistency and/or the prior matching information in discovering weakly consistent node pairs. 

To this end, we characterize and compute our own similarity measure between nodes across two networks, termed the {\em dual-perception similarity}, which is composed of the {\em multi-layer embedding similarity} and the {\em Tversky similarity}. We first calculate the similarity of multi-layer embeddings using graph neural networks (GNNs), where the weight-sharing technique is used to consistently generate hidden representations for each network and a newly designed layer-wise reconstruction loss is used to precisely capture multi-hop neighborhood structures during training. Then, we calculate the Tversky similarity, which is a newly devised asymmetric set similarity measure using the Tversky index \cite{tversky1977features} to alleviate the problem of network scale imbalance. By iteratively updating our dual-perception similarity, we gradually match a part of node pairs until all node correspondences are found.

\begin{figure} 
\begin{minipage}{.5\linewidth}
\centering
\subfloat[\label{1a}]{\label{main:a}\includegraphics[scale=.15]{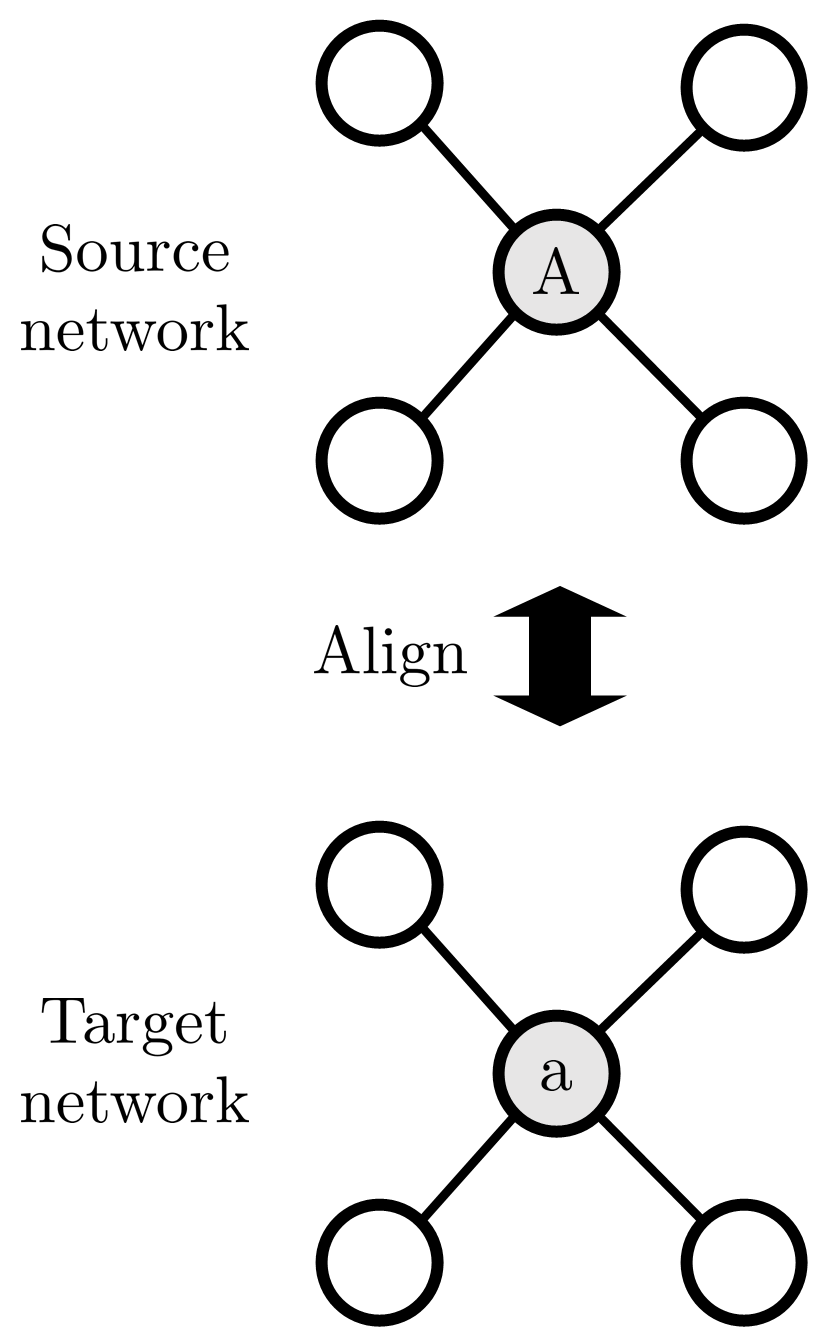}}
\end{minipage}%
\begin{minipage}{.5\linewidth}
\centering
\subfloat[\label{1b}]{\label{main:b}\includegraphics[scale=.15]{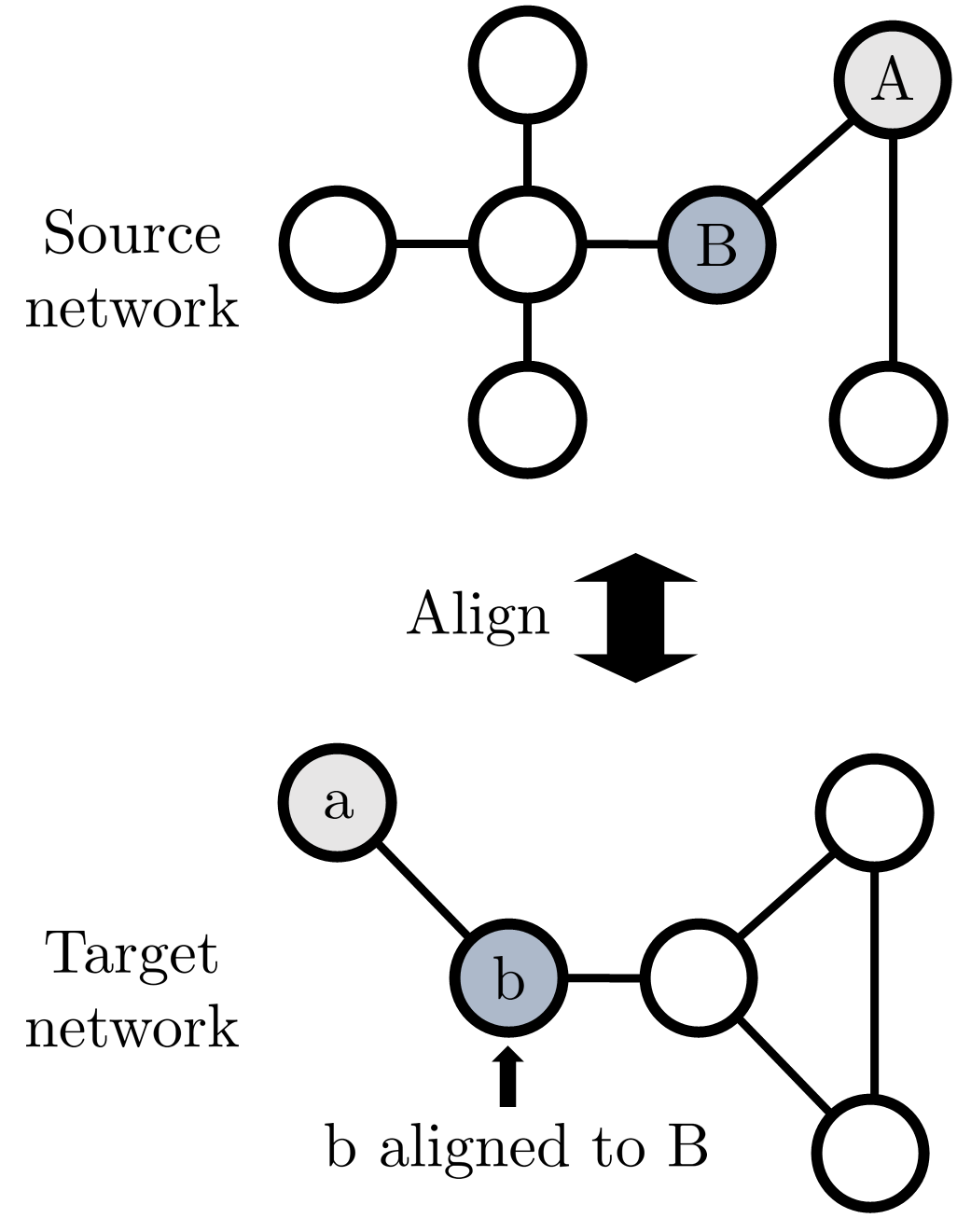}}
\end{minipage}\par\medskip
\caption{An example illustrating a node pair (A,a) that reveals (a) strong structural consistency and (b) weak structural consistency.}
\label{philosophy_fig_sub} 
\end{figure}

For a better understanding, we give an instance of two different networks including a node pair with either strong or weak structural consistency in Fig. \ref{philosophy_fig_sub}, which intuitively explains why gradually discovering node pairs is beneficial. Suppose that nodes A and B in a source network correspond to nodes a and b in a target network, respectively, as the ground truth mapping. Since a cross-network node pair (A, a) in Fig. \ref{1a} has strong structural consistency, it can be easily discovered. On the contrary, it is difficult to discover a node pair (A, a) in Fig. \ref{1b} due to its weak structural consistency. In this case, exploiting the information of already aligned node pairs that potentially exhibit strong consistency can be quite useful in discovering weakly consistent node pairs. For example, when we are aware of the node correspondence (B, b) beforehand through gradual matching, it would be much easier to discover the node pair (A, a) due to the fact that nodes A and a are direct neighbors of nodes B and b, respectively. Thus, we are capable of more effectively and precisely finding the node correspondence across different networks via such gradual matching.

The proposed design methodology is built upon rigorous theoretical frameworks by proving that 1) the weight-sharing technique in GNNs guarantees the consistency of cross-network nodes in the embedding space and 2) the impact and benefits of the Tversky similarity are higher than those of the well-known Jaccard index in terms of the growth rate of each similarity measure. Additionally, we analyze the computational complexity of our \textsf{Grad-Align} method. To further improve the performance of \textsf{Grad-Align} by strengthening the structural consistency, we also show its reinforced version that incorporates edge augmentation into the \textsf{Grad-Align} method.

To validate the superiority of our \textsf{Grad-Align} method, we comprehensively perform empirical evaluations using various real-world and synthetic datasets. Experimental results show that our method consistently outperforms state-of-the-art NA methods regardless of the datasets while showing substantial gains up to 75.54\% compared with the second-best performer. We also scrutinize the impact of each module in \textsf{Grad-Align} via ablation studies. Moreover, our experimental results demonstrate the robustness of our \textsf{Grad-Align} method to both structural and attribute noises owing to our dual-perception similarity.

The main technical contributions of this paper are five-fold and summarized as follows:
\begin{itemize}
    \item We introduce \textsf{Grad-Align}, a novel NA method that gradually finds the node correspondence;
    \item We characterize the dual-perception similarity to better capture the consistency of nodes across networks;
    \item We formulate the Tversky similarity applicable to networks with different scales;
    \item We validate the performance of \textsf{Grad-Align} through extensive experiments using real-world datasets as well as synthetic datasets;
    \item We further introduce \textsf{Grad-Align-EA}, a variant of the original \textsf{Grad-Align} method to reinforce the  structural consistency across networks.
\end{itemize}

\subsection{Organization and Notations}
The remainder of this paper is organized as follows. In Section \ref{section 2}, we present prior studies that are related to NA. In Section \ref{section 3}, we explain the methodology of our study, including the problem definition and an overview of our \textsf{Grad-Align} method. Section \ref{section 4} describes technical details of the proposed method. Comprehensive experimental results are shown in Section \ref{section 5}. Finally, we provide a summary and concluding remarks in Section \ref{section 6}.

Table \ref{NotationTable} summarizes the notation that is used in this paper. This notation will be formally defined in the following sections when we introduce our methodology and the technical details.

\begin{table}[t]
    \centering
    \begin{tabular*}{0.99\columnwidth}{ll}
        \toprule
        \textbf{Notation} & \textbf{Description} \\
        \midrule
        \rule{0pt}{7pt}$G_s$ & Source network\\
        \rule{0pt}{7pt}$\mathcal{V}_s$ & Set of nodes in $G_s$\\
        \rule{0pt}{7pt}$\mathcal{E}_s$ & Set of edges in $G_s$\\
        \rule{0pt}{7pt}$\mathcal{X}_s$ &Set of attributes of nodes in $\mathcal{V}_s$\\
        \rule{0pt}{7pt}$G_t$ & Target network\\
        \rule{0pt}{7pt}$\mathcal{V}_t$ & Set of nodes in $G_t$\\
        \rule{0pt}{7pt}$\mathcal{E}_t$ & Set of edges in $G_t$\\
        \rule{0pt}{7pt}$\mathcal{X}_t$ & Set of attributes of nodes in $\mathcal{V}_t$\\
        \rule{0pt}{7pt}$n_s$ & Number of nodes in $G_s$\\
        \rule{0pt}{7pt}$n_t$ & Number of nodes in $G_t$\\
        \rule{0pt}{7pt}$\pi^{(i)}$ & One-to-one node mapping at the $i$-th iteration\\
        \rule{0pt}{7pt}$M$ & Total number of ground truth node pairs\\
        \rule{0pt}{7pt}$\tilde{\mathcal{V}}_s^{(i)}$ & Set of seed nodes in $G_s$ up to the $i$-th iteration\\
        \rule{0pt}{7pt}$\tilde{\mathcal{V}}_t^{(i)}$ & Set of seed nodes in $G_t$ up to the $i$-th iteration\\
        \rule{0pt}{7pt}$\hat{\mathcal{V}}_s^{(i)}$ & Set of newly aligned nodes in $G_s$ at the $i$-th iteration\\
        \rule{0pt}{7pt}$\hat{\mathcal{V}}_t^{(i)}$ & Set of newly aligned nodes in $G_t$ at the $i$-th iteration\\
        \rule{0pt}{7pt}$\tilde{\mathcal{V}}_s^{(0)}$ & Set of prior seed nodes in $G_s$\\
        \rule{0pt}{7pt}$\tilde{\mathcal{V}}_t^{(0)}$ & Set of prior seed nodes in $G_t$\\
        \rule{0pt}{7pt}$\mathbf{H}_s^{(l)}$ & Hidden representation in $G_s$ at the $l$-th GNN layer\\
        \rule{0pt}{7pt}$\mathbf{H}_t^{(l)}$ & Hidden representation in $G_t$ at the $l$-th GNN layer\\
        \rule{0pt}{7pt}$\mathbf{S}_{emb}$ & Multi-layer embedding similarity matrix\\
        \rule{0pt}{7pt}$\mathbf{S}_{Tve}^{(i)}$ & Tversky similarity matrix at the $i$-th iteration\\
        \rule{0pt}{7pt}$\mathbf{S}^{(i)}$ & Dual-perception similarity matrix at the $i$-th iteration\\ 
        \bottomrule
    \end{tabular*}
    \caption{Summary of notations.}
    \label{NotationTable}
\end{table}

\section{Related Work}
\label{section 2}
The method that we propose in this paper is related to three broader fields of research, namely NA only with topological information, NA with attribute information, and network embedding-aided NA.

{\bf NA only with topological information.}
 The topological similarity of networks is a basic and core feature to identify the node correspondence between two given networks. Specifically, given the topological structure of two networks, IsoRank \cite{singh2008global} was designed by propagating the pairwise node similarity and structural consistency to discover the node correspondence. NetAlign \cite{bayati2009algorithms} formulated the NA problem as an integer quadratic program. 
 CLF \cite{zhang2015integrated} presented collective link fusion across partially aligned probabilistic networks. Moreover, BIG-ALIGN \cite{koutra2013big} was developed by using the alternating projected gradient descent approach that aims at solving a constrained optimization problem while finding permutation matrices in NA.

{\bf NA with attribute information.}
 In addition to the structural information, node attributes are another distinguishable feature that helps us find the node correspondence. For example, the profile attributes of each user such as the name, affiliation, and description can be valuable for aligning the same user across social networks \cite{zhang2018mego2vec}. REGAL \cite{heimann2018regal} was presented by performing a low-rank implicit approximation of a similarity matrix that incorporates the structural similarity and attribute agreement between nodes in two disjoint graphs. ULink \cite{mu2016user} was proposed by exploring the concept of latent userspace to more naturally model the relationship between the underlying real users. FINAL \cite{zhang2016final} was presented by leveraging not only the node attribute information but also the edge attribute information in the topology-based NA process.
 
{\bf Network embedding-aided NA.}
With the increasing attention to network embedding (also known as network representation learning) and its diverse applications in solving downstream ML problems, network embedding-aided NA has recently become in the spotlight. Network embedding learns a mapping from each node in a graph to a low-dimensional vector in an embedding space while preserving intrinsic network properties (e.g., the neighborhood structure and high-order proximities), thus resulting in an efficient and scalable representation of the underlying graph \cite{cui2018survey}. PALE \cite{man2016predict} was presented by exploiting the first-order and second-order proximities of node pairs in the embedding space and further adopting a multi-layer perceptron (MLP) architecture to capture the nonlinear relationship between the resulting network embeddings from two different networks. DeepLink \cite{zhou2018deeplink} was presented by employing an unbiased random walk to generate node embeddings of two given networks based on the Skip-gram model \cite{mikolov2013distributed} and using an autoencoder as a mapping function to discover the relation between two embeddings. IONE \cite{liu2016aligning} was designed for solving the NA problem by learning an aligned network embedding in multiple directed and weighted networks. CENALP \cite{du2019joint} was presented by jointly performing NA and link prediction tasks to enhance the alignment accuracy, where a cross-graph embedding method based on random walks was devised. In \cite{du2019joint}, the structural and node attribute similarities were taken into account to predict cross-network links. 
Moreover, since GNNs \cite{DBLP:conf/iclr/KipfW17, DBLP:conf/nips/HamiltonYL17,xu2018powerful,velivckovic2017graph} have  emerged as a powerful network feature extractor in attributed networks, GAlign \cite{trung2020adaptive} was developed by making use of the multi-order (multi-layer) nature of graph convolutional network (GCN) \cite{DBLP:conf/iclr/KipfW17} for NA. 

{\bf Discussion.} Although the aforementioned NA approaches achieve convincing alignment performance under their own network settings, they pose several practical challenges. 
Precisely, the methods were inherently designed in such a way that their performance depend highly on either the topological information \cite{koutra2013big, singh2008global,zhou2018deeplink,zhang2015integrated} or the attribute information \cite{trung2020adaptive, zhang2016final,heimann2018regal}; such a high dependency makes the designed model vulnerable to topological or attribute inconsistency across networks. Moreover, most of the conventional methods such as \cite{koutra2013big,singh2008global,zhou2018deeplink,heimann2018regal,tan2014mapping,zhang2015integrated} find all node pairs {\em at once} without leveraging already discovered node pairs during the node matching, which often fails to correctly find the correspondence of nodes exhibiting weak consistency. Although there was an attempt to find node pairs iteratively (see \cite{du2019joint}), the method focuses on performing NA along with link prediction to enrich the structure information by newly added links. Thus, the impact and benefits of gradual alignment in improving the alignment accuracy were underexplored yet. 


\begin{figure*}
    \centering
    \includegraphics[width=1.0\textwidth]{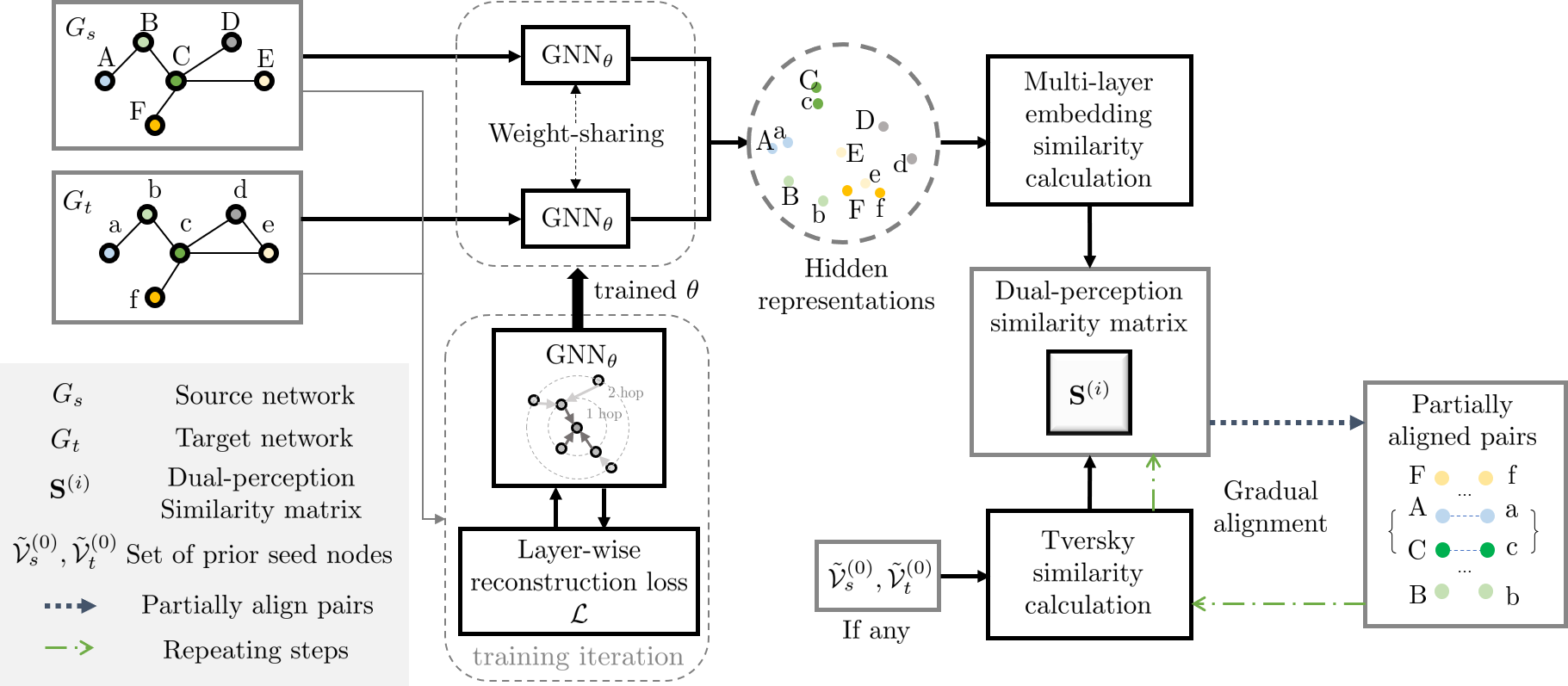}
    \caption{The schematic overview of our \textsf{Grad-Align} method.}
    \label{fig:overview}
\end{figure*}
\section{Methodology}
\label{section 3}
In this section, as a basis for the proposed \textsf{Grad-Align} method in Section \ref{section 4}, we first describe our network model with basic assumptions and formulate the NA problem. Then, we explain an overview of our \textsf{Grad-Align} method using dual-perception similarity measures as a solution to the problem.

\subsection{Network Model and Basic Assumptions}
\label{sub3.1}
We assume source and target networks to be aligned, denoted as $G_s$ and $G_t$, respectively. For simplicity, we assume both $G_s$ and $G_t$ to be undirected and unweighted networks without self-loop or repeated edges. For the source network $G_s = (\mathcal{V}_s,\mathcal{E}_s,\mathcal{X}_s)$, $\mathcal{V}_s$ is the set of nodes (or equivalently, vertices) in $G_s$ whose size is $n_s$ (i.e., $|\mathcal{V}_s|=n_s$); $\mathcal{E}_s$ is the set of edges between pairs of nodes in $\mathcal{V}_s$; and $\mathcal{X}_s \in \mathbb{R}^{n_s \times d}$ is the set of node attribute (feature) vectors, where $d$ is the number of attributes per node. For the target network $G_t = (\mathcal{V}_t,\mathcal{E}_t,\mathcal{X}_t)$, notations $\mathcal{V}_t$, $\mathcal{E}_t$, and $\mathcal{X}_t$ similarly follow those in $G_s$ with $|\mathcal{V}_t|=n_t$ and $\mathcal{X}_t \in \mathbb{R}^{n_t \times d}$.\footnote{Here, the number of node attributes for both networks is assumed to be the same as in other NA methods (see \cite{trung2020adaptive,zhang2016final,du2019joint}).} In our study, we assume that there are $M$ ground truth cross-network node pairs to be discovered.

\subsection{Problem Definition}

    In the following, we formally define the NA problem for given two networks $G_s$ and $G_t$ as follows.

    \begin{definition}[NA] Given two networks
    $G_s=(\mathcal{V}_s,\mathcal{E}_s,\mathcal{X}_s)$ and $G_t=(\mathcal{V}_t,\mathcal{E}_t,\mathcal{X}_t)$, NA aims to find one-to-one mapping $\pi:\mathcal{V}_s \rightarrow \mathcal{V}_t$, where $\pi(u) = v$ and $\pi^{-1}(v) = u$ for $u \in \mathcal{V}_s$ and $v \in \mathcal{V}_t$.
    \end{definition}



\subsection{Overview of \textsf{Grad-Align}}

In this subsection, we describe the overview of our \textsf{Grad-Align} method using the {\em dual-perception} similarity. We start by stating that the standard protocol of NA involves the computation of the similarity between cross-network node pairs, followed by the utilization of assignment algorithms to discover cross-network node pairs from the similarity matrix \cite{zhang2016final, heimann2018regal,trung2020adaptive,kollias2011network}. Our method basically follows the standard protocol of NA but attempts to gradually discover cross-network node pairs based on two different types of similarity measures, including 1) the similarity of {\it multi-layer} embeddings and 2) the {\it Tversky similarity} as a newly characterized asymmetric set similarity using the Tversky index \cite{tversky1977features}. First, to generate multi-layer embeddings, GNNs are adopted for effectiveness in representing the structural and attribute semantic relations between nodes. Second, the Tversky similarity enables us to overcome the problem raised by networks with different scales.
As illustrated in Fig. \ref{fig:overview}, our \textsf{Grad-Align} is basically composed of three main phases: 1) the multi-layer embedding similarity calculation, 2) the Tversky similarity calculation, and 3) the gradual node matching.

(Phase 1) We first focus on calculating the {\em multi-layer} embedding similarity matrix via $L$-layer GNN. Let $\mathbf{H}_s^{(l)}\in \mathbb{R}^{n_s \times h}$ and $\mathbf{H}_t^{(l)}\in \mathbb{R}^{n_t \times h}$ denote the hidden representation in $G_s$ and $G_t$, respectively, at the $l$-th GNN layer where $l\in\{1,\cdots,L\}$ and $h$ is the dimension of each representation vector. As illustrated in Fig. \ref{fig:overview}, in the feed-forward process of GNN, we use the weight-sharing technique in order to consistently generate $\mathbf{H}_s^{(l)}$ and $\mathbf{H}_t^{(l)}$ at each layer. We then train the generated GNN model parameters via a {\em layer-wise reconstruction loss} to make each node representation more distinguishable by exploiting the nodes' proximities up to the $l$-th order (which will be specified in Section \ref{sec 4.1}). Using the hidden representations $\mathbf{H}_s^{(l)}$ and $\mathbf{H}_t^{(l)}$ at each layer through training iterations, we are capable of computing the multi-layer embedding similarity matrix ${\bf S}_{emb}\in \mathbb{R}^{n_s\times n_t}$, which is expressed as
    \begin{equation}
    \label{emb_sim}
    \mathbf{S}_{emb} = \sum_l\mathbf{H}_s^{(l)}{\mathbf{H}_t^{(l)}}^{\top},
    \end{equation}
where the superscript $\top$ denotes the transpose of a matrix and each entry in ${\bf S}_{emb}$ represents the similarity of pairwise node embedding vectors in the two networks $G_s$ and $G_t$.

(Phase 2) We turn to calculating the Tversky similarity matrix. We start by denoting the set of one-hop neighbors of node $u$ in $G_s$ and the set of one-hop neighbors of node $v$ in $G_t$ as $\mathcal{N}_{G_s, u}$ and $\mathcal{N}_{G_t, v}$, respectively. Since the Tversky similarity calculation runs iteratively, we define $\pi^{(i)}$ as the mapping function at the $i$-th iteration. To characterize our new similarity, we also define the ``aligned cross-network neighbor-pair (ACN)" as follows.
\begin{definition}
(ACN). Given a node pair $(u,v)$, if $x\in\mathcal{N}_{G_s,u}$, $y\in\mathcal{N}_{G_t,v}$, and $\pi^{(i)}(x)=y$ (i.e., $(x,y)$ is the already aligned node pair), then the node pair $(x,y)$ belongs an ACN of $(u,v)$ across two networks.
\end{definition}
Then, for each iteration, we are interested in counting the number of ACNs of node pair $(u,v)$, i.e., the intersection of the two sets $\mathcal{N}_{G_s, u}$ and $\mathcal{N}_{G_t, v}$, indicating one-hop neighbors of nodes $u\in\mathcal{V}_s$ and $v\in\mathcal{V}_t$, respectively, upon the one-to-one node mapping $\pi^{(i)}$. However, there is a practical challenge on measuring ACNs. In particular, since two networks are often severely imbalanced in terms of the number of nodes, the most well-known set similarity measure, referred to as the Jaccard index, fails to precisely capture what portion of ACNs are shared. To remedy this problem, we present the Tversky similarity as an asymmetric set similarity measure using the Tversky index. We compute the Tversky similarity matrix ${\bf S}_{Tve}^{(i)}\in \mathbb{R}^{n_s\times n_t}$ iteratively by balancing between the set differences of two sets. The detailed description of ${\bf S}_{Tve}^{(i)}$ will be shown in Section \ref{section 4}. Finally, we calculate our proposed similarity matrix ${\bf S}^{(i)}\in \mathbb{R}^{n_s\times n_t}$, namely the dual-perception similarity matrix, at each iteration using the two similarity matrices as follows:
    \begin{equation}
    \label{final sim}
    \mathbf{S}^{(i)}=\mathbf{S}_{emb}\odot \mathbf{S}^{(i)}_{Tve},
    \end{equation}
where $\odot$ indicates the element-wise matrix multiplication operator.

(Phase 3) We explain how to gradually match node pairs using our dual-perception similarity ${\bf S}^{(i)}$ in (\ref{final sim}) for each iteration. Since ${\bf S}_{emb}$ does not change over iterations, we focus only on how to update the Tversky similarity ${\bf S}_{Tve}^{(i)}$. From the fact that the NA problem is often solved under the supervision of some observed anchor nodes \cite{zhang2016final,trung2020adaptive,man2016predict}, we suppose that there are two sets of {\em prior seed nodes} (also known as anchor nodes) in two different networks, denoted by $\tilde{\mathcal{V}}_s^{(0)}$ and $\tilde{\mathcal{V}}_t^{(0)}$ in source and target networks, respectively, whose links connecting them correspond to the ground truth prior anchor information for NA. These two sets are utilized to \textit{initially} calculate ${\bf S}_{Tve}^{(1)}$ along with the node mapping function $\pi^{(0)}$.\footnote{If $\tilde{\mathcal{V}}_s^{(0)}$ and $\tilde{\mathcal{V}}_t^{(0)}$ are $\emptyset$, which means that no prior information is given, then we only use $\mathbf{S}_{emb}$ at the first matching step.} By iteratively updating ${\bf S}_{Tve}^{(i)}$ based on the information of newly aligned node pairs, \textsf{Grad-Align} makes full use of node pairs exhibiting strong consistency as well as prior seed node pairs (i.e., pairs that connect $\tilde{\mathcal{V}}_s^{(0)}$ and $\tilde{\mathcal{V}}_t^{(0)}$). Next, we explain our gradual alignment process. Let $\tilde{\mathcal{V}}_s^{(i)}$ and $\tilde{\mathcal{V}}_t^{(i)}$ denote the set of aligned nodes in $G_s$ and $G_t$, respectively, up to the $i$-th iteration. Then, it follows that $\tilde{\mathcal{V}}_*^{(i)}=\tilde{\mathcal{V}}_*^{(i-1)} \cup \hat{\mathcal{V}}_*^{(i)}$, where the subscript * represents $s$ and $t$ for source and target networks, respectively, and $\hat{\mathcal{V}}_*^{(i)}$ is the set of {\em newly} aligned nodes in each network at the $i$-th iteration. By iteratively calculating the dual-perception similarity matrix ${\bf S}^{(i)}$ in (\ref{final sim}), we discover $|\hat{\mathcal{V}}_s^{(i)}|=|\hat{\mathcal{V}}_t^{(i)}|=N$ cross-network node pairs at each iteration, where $|\cdot|$ is the cardinality of the set, and $N>0$ is a positive integer. To this end, similarly as in \cite{du2019joint, zhang2016final}, we employ a ranking-based selection strategy that selects the top-$N$ elements in the matrix ${\bf S}^{(i)}$ for each iteration (which will be specified in Section \ref{sec 4.1}). Then, since the node mapping $\pi^{(i)}$ is updated to $\pi^{(i+1)}$ after $N$ node pairs are newly found, the matrix ${\bf S}_{Tve}^{(i+1)}$ is recalculated accordingly in order to update $\pi^{(i+2)}$. The process is repeated until all node correspondences (i.e., $M$ node pairs) are found. Specifically, the iteration ends when $i$ becomes $\ceil*{\frac{M}{N}}+1$, where $\ceil{\cdot}$ is the ceiling operator.\footnote{Note that the number of node pairs to be aligned at the last iteration may be less than $N$.}

The schematic overview of our \textsf{Grad-Align} is presented in Fig. \ref{fig:overview}. The two most consistent pairs (A, a) and (C, c) are aligned at the first step as the similarities of (A, a) and (C, c) in ${\bf S}^{(1)}$ are higher than those of other cross-network node pairs.

\section{Proposed \textsf{Grad-Align} Method}\label{section 4}
In this section, we elaborate on our \textsf{Grad-Align} method that gradually discovers the node correspondence using dual-perception similarities. We also analyze the computational complexity of \textsf{Grad-Align}. Furthermore, we present an enhanced version of \textsf{Grad-Align}, namely \textsf{Grad-Align-EA}, that incorporates the edge augmentation module into the original method.

\label{sec 4.1}
\begin{algorithm}[t]
\caption{: \textsf{Grad-Align}}
\label{mainalgorithm}
 \begin{algorithmic}[1]
  \renewcommand{\algorithmicrequire}{\textbf{Input:}}
  \renewcommand{\algorithmicensure}{\textbf{Output:}}
  \REQUIRE $G_s$, $G_t$, $\tilde{\mathcal{V}}_s^{(0)}$ ,$\tilde{\mathcal{V}}_t^{(0)}$
  \ENSURE $\pi^{(\ceil{\frac{M}{N}}+1)}$
  \STATE \textbf{Initialization: } $\theta \leftarrow \text{random initialization}$; $i \leftarrow 1$
  \STATE /* Calculation of ${\bf S}_{emb}$ */
   \WHILE {not converged}
  \STATE  $\mathbf{H}^{(1)}_s, \mathbf{H}^{(2)}_s, ..., \mathbf{H}^{(L)}_s\leftarrow \textsf{GNN}_{\theta}(G_s)$
  \STATE  $\mathbf{H}^{(1)}_t, \mathbf{H}^{(2)}_t, ..., \mathbf{H}^{(L)}_t\leftarrow \textsf{GNN}_{\theta}(G_t)$
  \STATE $\mathcal{L} \leftarrow$ \text{layer-wise reconstruction loss} in (\ref{layer-wise reconstruction loss})
  \STATE Update $\theta$ by taking one step of gradient descent
  \ENDWHILE
  \STATE $\mathbf{S}_{emb} \leftarrow  \sum^L_{l=1}\mathbf{H}_s^{(l)}{\mathbf{H}_t^{(l)}}^{\top}$
  \STATE /* Gradual matching by updating ${\bf S}_{Tve}^{(i)}$ */
  \IF {$\tilde{\mathcal{V}}_s^{(0)}=\tilde{\mathcal{V}}_t^{(0)}=\emptyset$}
    \STATE $\mathbf{S}^{(i)}\leftarrow \mathbf{S}_{emb}$
  \ELSE 
    \STATE Calculate $\mathbf{S}_{Tve}^{(i)}$ using (\ref{Tversky sim})
    \STATE $\mathbf{S}^{(i)}\leftarrow \mathbf{S}_{emb}\odot \mathbf{S}_{Tve}^{(i)}$
  \ENDIF
  
  \WHILE {$i\le \ceil{\frac{M}{N}}$} 
  \STATE Find $\hat{\mathcal{V}}_s^{(i)}$ and $\hat{\mathcal{V}}_t^{(i)}$ based on $\mathbf{S}^{(i)}$
  \STATE $i \leftarrow i + 1$
  \STATE Update mapping $\pi^{(i)}$: \\$\{ \pi^{(i)}(u)=v \mid u \in \hat{\mathcal{V}}_s^{(i-1)}, v \in \hat{\mathcal{V}}_t^{(i-1)}\}$
  \STATE Update ${\bf S}_{Tve}^{(i)}$ using (\ref{Tversky sim})
  \STATE ${\bf S}^{(i)} \leftarrow {\bf S}_{emb} \odot {\bf S}_{Tve}^{(i)}$

  \ENDWHILE
  \RETURN $\pi^{(\ceil{\frac{M}{N}}+1)}$
 \end{algorithmic}
\end{algorithm}

\subsection{Methodological Details of \textsf{Grad-Align}}
\label{sec 4.1}
We start by recalling that \textsf{Grad-Align} consists of three key phases, including 1) the multi-layer embedding similarity calculation, 2) the Tversky similarity calculation, and 3) the gradual node matching. The overall procedure of the proposed \textsf{Grad-Align} method is summarized in Algorithm \ref{mainalgorithm}, where $\theta$ and $\textsf{GNN}_\theta$ indicate the trainable model parameters and the GNN model parameterized by $\theta$.

First, let us describe how to calculate the multi-layer embedding similarity matrix ${\bf S}_{emb}$ using a GNN model. Since our \textsf{Grad-Align} method is {\em GNN-model-agnostic}, we show a general form of the message passing mechanism~\cite{gilmer2017neural,DBLP:conf/nips/HamiltonYL17,xu2018powerful} of GNNs in which we repeatedly update the representation of each node by aggregating representations of its neighbors using two functions, namely AGGREGATE and UPDATE. Formally, at the $l$-th hidden layer of a GNNs, $\text{AGGREGATE}^{l}_{u}$ aggregates (latent) feature information from the local neighborhood of node $u$ in a given graph as follows:
\begin{equation}
\label{aggregateequation}
    \mathbf{m}^{l}_{u} \leftarrow \text{AGGREGATE}^{l}_{u}(\{\mathbf{h}_{x_u}^{l-1}|x_u \in \mathcal{N}_{u} \cup u\}, \mathbf{W}^l_1),
\end{equation}
where ${\bf h}^{l-1}_{x_u} \in \mathbb{R}^{1 \times h}$ denotes the $h$-dimensional latent representation vector of node $x_u$ at the $(l-1)$-th hidden layer; $\mathcal{N}_{u}$ indicates the set of neighbor nodes of $u$; $\mathbf{W}^l_1 \in \mathbb{R}^{h \times h}$ is the learnable weight matrix at the $l$-th layer for the AGGREGATE step; and ${\bf m}_u^l$ is the aggregated information at the $l$-th layer. Here, ${\bf h}_u^0$ represents the node attribute vector of node $u$. In the UPDATE step, the latent representation at the next hidden layer is updated by using each node and its aggregated information from $\text{AGGREGATE}^{l}_{u}$ as follows:
    \begin{equation}
    \label{UPDATEequation}
        \mathbf{h}^{l}_{u} \leftarrow \text{UPDATE}^{l}_{u}(u,\mathbf{m}^{l}_{u},\mathbf{W}^l_2),
    \end{equation}
where $\mathbf{W}^l_2 \in \mathbb{R}^{h^m \times h^{l}}$ is the learnable weight matrix at the $l$-th layer for the UPDATE step.\footnote{Note that the two learnable weight matrices ${\bf W}_1^l$ and ${\bf W}_2^l$ can be excluded depending on a specific choice of GNNs.} The above two functions AGGREGATE and UPDATE in (\ref{aggregateequation}) and (\ref{UPDATEequation}), respectively, can be specified by several milestone GNN models such as GCN \cite{DBLP:conf/iclr/KipfW17}, GraphSAGE \cite{DBLP:conf/nips/HamiltonYL17}, and GIN \cite{xu2018powerful}. Meanwhile, there is a practical challenge on implementing GNN model architectures aimed at performing the NA task.  Specifically, since most of existing GNN models are not originally designed for learning node representations in multiple networks, a direct application of GNNs to the two networks $G_s$ and $G_t$ may not generate desirable node embeddings along with structural and attribute consistencies of nodes \cite{trung2020adaptive}. In other words, while a pair of nodes across two networks has the same neighborhood structure and/or attribute information, their representations may fail to be mapped closely together into the embedding space. To solve this problem, we employ the technique of sharing weight matrices between two GNN models for $G_s$ and $G_t$ (refer to lines 4--5 in Algorithm \ref{mainalgorithm}). We would like to establish the following proposition, which states that the weight-sharing technique in GNNs naturally supports the consistency of cross-network nodes in the embedding space in the context of NA.
    \begin{proposition}
    \label{proposition_weightshare}
     Suppose that two networks $G_s$ and $G_t$ are isomorphic,\footnote{Two networks $G_s$ and $G_t$ are isomorphic when there exists a function $f: \mathcal{V}_s \rightarrow \mathcal{V}_t$ such that any two nodes $u_1$ and $u_2$ of $\mathcal{V}_s$ are adjacent in $G_s$ if and only if $f(u_1)$ and $f(u_2)$ are adjacent in $G_t$ \cite{mckay2014practical}.} where any of node pairs $(u,v)$ matched by ground truth node mapping $\pi: \mathcal{V}_s \rightarrow \mathcal{V}_t$ have the same node attributes $\mathbf{h}^{0}_{u} = \mathbf{h}^{0}_{v}$. Given two GNN models for $G_s$ and $G_t$, if the weight matrices $\{\mathbf{W}^{p}_1\}_{p=1,\cdots,L}$ and $\{\mathbf{W}^{p}_2\}_{p=1,\cdots,L}$ are shared, then it follows that $\mathbf{h}^{p}_{u} = \mathbf{h}^{p}_{v}$ for all positive integers $p$, where $L$ indicates the number of GNN layers.
    \end{proposition}
    \begin{proof}
    We prove this proposition by the mathematical induction for all layers $p=1,\cdots,L$, where $p=1$ and $p>2$ correspond to the base step and the consecutive inductive steps, respectively.
    \\{\bf Base step:} For node $u\in\mathcal{V}_s$, it follows that
    \begin{equation}
    \label{GNNs_1}
    \begin{aligned}
            \mathbf{m}^{1}_{u} &\leftarrow \text{AGGREGATE}^{1}_{u}(\{\mathbf{h}_{x}^{0}|x \in \mathcal{N}_{u} \cup u\}, \mathbf{W}^1_1)\\
        \mathbf{h}^{1}_{u} &\leftarrow \text{UPDATE}^{1}_{u}(u,\mathbf{m}^{1}_{u},\mathbf{W}^1_2).
    \end{aligned}
    \end{equation}
    For node $v\in\mathcal{V}_t$, we have
    \begin{equation}
    \label{GNNt_1}
    \begin{aligned}
            \mathbf{m}^{1}_{v} &\leftarrow \text{AGGREGATE}^{1}_{v}(\{\mathbf{h}_{y}^{0}|y \in \mathcal{N}_{v} \cup v\}, \mathbf{W}^1_1)\\
        \mathbf{h}^{1}_{v} &\leftarrow \text{UPDATE}^{1}_{v}(v,\mathbf{m}^{1}_{v},\mathbf{W}^1_2).
    \end{aligned}
    \end{equation}
     From (\ref{GNNs_1}) and (\ref{GNNt_1}), the equality $\{\mathbf{h}_{x}^{0}|x \in \mathcal{N}_{u} \cup u\}=\{\mathbf{h}_{y}^{0}|y \in \mathcal{N}_{v} \cup v\}$ is met due to the assumptions that $G_s$ and $G_t$ are isomorphic and ${\bf h}_u^0={\bf h}_v^0$, thus resulting in ${\bf m}_u^1={\bf m}_v^1$ from the AGGREGATE function. Using the UPDATE function finally yields ${\bf h}_u^1={\bf h}_v^1$.
    \\{\bf Inductive step:} Suppose that
    \begin{equation}
    \label{suppose_gnn}
        \mathbf{h}^{p}_{u} = \mathbf{h}^{p}_{v}
    \end{equation}
    for $p=k$, where $k\in\{1,\cdots,L-1\}$. Then, it is not difficult to show that (\ref{suppose_gnn}) also holds for $p=k+1$ using (\ref{aggregateequation}) and (\ref{UPDATEequation}), which completes the proof of this proposition. 
    \end{proof}

Now, we describe how to train GNN model parameters $\theta$. Unlike prior studies including \cite{man2016predict, zhou2018deeplink}, which utilized supervision data to learn the relationship between distinct latent spaces, we use the weight-sharing technique in GNNs that inherently supports the node consistency across two different networks in the embedding space (refer to Proposition \ref{proposition_weightshare}). Thus, in our study, we train the parameters of GNNs unsupervisedly, similarly as in \cite{trung2020adaptive}. From the fact that the hidden representation at the $l$-th GNN layer contains the collective information of up to $l$-hop neighbors of nodes \cite{zhou2020graph, zhang2020deep, DBLP:conf/iclr/KipfW17, DBLP:conf/nips/HamiltonYL17}, we make use of the adjacency matrix and its powers up to order $l$ along with the $l$-th layer hidden representation. To this end, we present our new {\em layer-wise reconstruction loss} function for training the GNN model as follows (refer to line 6):
\begin{equation}
\label{layer-wise reconstruction loss}
    \mathcal{L} = \sum_{* \in \{s,t\}}\sum_l \left\|\tilde{D}_*^{(l)-\frac{1}{2}}\tilde{A}_*^{(l)}\tilde{D}_*^{(l)-\frac{1}{2}}-H_*^{(l)}{H_*^{(l)}}^{\top}\right\|_F,
\end{equation}
where the subscript * represents $s$ and $t$ for source and target networks, respectively; $\tilde{A}_*^{(l)} =\sum_{k=1}^l\hat{A}_*^{k}$ where $\hat{A}_*=A_*+I_*$ is the adjacency matrix with self-connections in which $I_*$ is the identity matrix; $\tilde{D}_*^{(l)}$ is a diagonal matrix whose $(i,i)$-th element is $\tilde{D}_*^{(l)}(i,i) = \sum_j\tilde{A}_*^{(l)}(i,j)$ where $\tilde{A}_*^{(l)}(i,j)$ is the $(i,j)$-th element of $\tilde{A}_*^{(l)}$; and $\|\cdot\|_F$ is the Frobenius norm of a matrix. By training the GNN model via our layer-wise reconstruction loss in (\ref{layer-wise reconstruction loss}), we are capable of precisely capturing the first-order and higher-order neighborhood structures. Note that, unlike the prior study (e.g., \cite{trung2020adaptive}), we use the {\em aggregated} matrix $\tilde{A}_*^{(l)}$ to reconstruct the underlying network from each hidden representation. 

Next, we turn to explaining how to calculate the Tversky similarity matrix ${\bf S}_{Tve}^{(i)}$. Without loss of generality, we assume that $n_s \ge n_t$ for the Tversky similarity calculation below. Then, given a cross-network node pair $(u,v)$, we evaluate how many ACNs are shared between two sets $\mathcal{N}_{G_s,u}$ and $\mathcal{N}_{G_t,v}$, indicating one-hop neighbors of nodes $u\in\mathcal{V}_s$ and $v\in\mathcal{V}_t$, respectively, upon the one-to-one node mapping $\pi^{(i)}$. To this end, we characterize the Tversky similarity ${\bf S}_{Tve}^{(i)}(u,v)$ between two nodes $u\in\mathcal{V}_s$ and $v\in\mathcal{V}_t$ at the $i$-th iteration by means of the mapping $\pi^{(i)}$, which is formulated as:
\begin{equation}
\begin{aligned}
\label{Tversky sim}
&\mathbf{S}_{Tve}^{(i)}(u,v) \\&= \frac{|X_u^{(i)} \cap Y_v^{(i)}|}{|X_u^{(i)} \cap Y_v^{(i)}|+\alpha|X_u^{(i)}-Y_v^{(i)}|+\beta|Y_v^{(i)}-X_u^{(i)}|}, 
\end{aligned}
\end{equation}
where 
\begin{equation}
\begin{aligned}
\label{Tversky sim comp}
\mathcal{T}_u^{(i)} &= \left\{\pi^{(i)}(x) \mid x \in \left(\mathcal{N}_{G_s,u} \cap \tilde{\mathcal{V}}_s^{(i)}\right)\right\} \\
X_u^{(i)} &= \left(\mathcal{N}_{G_s,u}-\tilde{\mathcal{V}}_s^{(i)}\right) \cup \mathcal{T}_u^{(i)}\\
Y_v^{(i)} &= \mathcal{N}_{G_t,v}. \\
\end{aligned}    
\end{equation}
Here, $\tilde{\mathcal{V}}_s^{(i)}$ is the set of {\em seed nodes} at the $i$-th iteration of gradual alignment in $G_s$, which is defined as the union of already aligned node pairs up to the $i$-th iteration and $\tilde{\mathcal{V}}_s^{(0)}$ (i.e., prior seed nodes in $G_s$); and $\alpha,\beta>0$ are the parameters of the Tversky index. In our study, we set $0<\alpha<1$ and $\beta=1$ under the condition $n_s\ge n_t$, where the appropriate value of $\alpha$ will be numerically found in Section \ref{sec5.5.2}. It is worth noting that the case of $\alpha = \beta = 1$ corresponds to the Tanimoto coefficient (also known as the Jaccard index) \cite{tversky1977features}, which can be seen as a symmetric set similarity measure. In particular, we would like to address the importance of the parameter setting in our Tversky similarity. 
\begin{remark}
The parameter $\alpha$ in (\ref{Tversky sim}) plays a crucial role in balancing between $|X_u^{(i)} \cap Y_v^{(i)}|$ (i.e., the number of ACNs of node pair $(u,v)$) and $|X_u^{(i)} - Y_v^{(i)}|$ (i.e., the  cardinality of the set difference of $X_u^{(i)}$ and $Y_v^{(i)}$). More specially, for
$n_s \gg n_t$, a number of nodes in $G_s$ remain unaligned even when the alignment process is completed; thus, the term $|X_u^{(i)} \cap Y_v^{(i)}|$ would be far smaller than $|X_u^{(i)} - Y_v^{(i)}|$.\footnote{We emprically validated that over 99\% out of all node pairs meet this condition (i.e., $|X_u^{(i)} \cap Y_v^{(i)}| \ll |X_u^{(i)} - Y_v^{(i)}|$) in the gradual alignment process in three real-world datasets under consideration.} In this case, we are capable of more precisely calculating the similarity between node pairs of two networks that are severely imbalanced during the gradual alignment by adjusting the value of $\alpha$ accordingly. By setting $0<\alpha<1$, we can mitigate the impact of $|X_u^{(i)}-Y_v^{(i)}|$ so that $|X_u^{(i)} \cap Y_v^{(i)}|$ is relatively influential.
\end{remark} 
Since the number of ACNs is expected to increase over iterations in gradual alignment, we are interested in investigating how fast the Tversky similarity ${\bf S}_{Tve}^{(i)}(u,v)$ grows with respect to the number of ACNs. To see the impact and benefits of the Tversky similarity in terms of its growth rate, we analytically show the effectiveness of the Tversky similarity with a proper parameter setting by establishing the following theorem. 
\begin{theorem} 
\label{theorem_tversky}
Given a node pair $(u,v)$, suppose that $|Y_v^{(i)}-X_u^{(i)}| \ge |X_u^{(i)} \cap Y_v^{(i)}|$ for all gradual alignment steps. Then, under the condition $n_s\ge n_t$, the growth rate of our Tversky similarity ${\bf S}_{Tve}^{(i)}(u,v)$ using $0<\alpha<1$ and $\beta=1$ with respect to $|X_u^{(i)} \cap Y_v^{(i)}|$ (i.e., the number of ACNs of node pair $(u,v)$ at the $i$-th iteration) is always higher than that of the Jaccard index using $\alpha=\beta=1$.
\end{theorem}
\begin{proof}
Let $\mathbf{S}^{(i)}_{Jac}(u,v)$ denote the Jaccard index between two nodes $u\in\mathcal{V}_s$ and $v\in\mathcal{V}_t$ in the $i$-th iteration. For notational convenience, by denoting $a=|X_u^{(i)}-Y_v^{(i)}|$, $b=|Y_v^{(i)}-X_u^{(i)}|$, and $c=|X_u^{(i)} \cap Y_v^{(i)}|$ for $b \ge c \ge 0$, we have
\begin{equation}
\begin{aligned}
\label{lemma eq1}
    \mathbf{S}^{(i)}_{Tve}(u,v) &= \frac{c}{\alpha a+b+c} = 1 - \frac{\alpha a+b}{\alpha a+b+c} \\
    \mathbf{S}^{(i)}_{Jac}(u,v) &= \frac{c}{a+b+c} = 1 - \frac{a+b}{a+b+c}.
\end{aligned}
\end{equation}

Then, the difference between the growth rates of ${\bf S}_{Tve}^{(i)}(u,v)$ and ${\bf S}_{Jac}^{(i)}(u,v)$ with respect to $c$ is given by
\begin{equation}
\begin{aligned}
\label{lemma eq3}
    \frac{\partial{\mathbf{S}^{(i)}_{Jac}(u,v)}}{\partial{c}} - \frac{\partial{\mathbf{S}^{(i)}_{Tve}(u,v)}}{\partial{c}} = \frac{w_1}{(w_1+c)^2} - \frac{w_2}{(w_2+c)^2} \\ = \frac{(w_1-w_2)(c^2-w_1w_2)}{(w_1+c)^2(w_2+c)^2},
\end{aligned}
\end{equation}
where $w_1=a+b$ and $w_2=\alpha a+b$. Here, due to the fact that $w_1-w_2=(1-\alpha)a \ge 0$, $w_1 \ge c$, and $w_2 \ge c$, it follows that
\begin{equation}
\begin{aligned}
\label{lemma eq4}
    \frac{\partial{\mathbf{S}^{(i)}_{Tve}(u,v)}}{\partial{c}} \ge \frac{\partial{\mathbf{S}^{(i)}_{Jac}(u,v)}}{\partial{c}},
\end{aligned}
\end{equation}
which completes the proof of this theorem.
\end{proof}
By Theorem \ref{theorem_tversky}, using the parameter setting $0<\alpha<1$ and $\beta=1$ in the Tversky similarity facilitates the increase of the similarity level during the gradual alignment, thus better capturing the topological consistency. In other words, we can make full use of the benefits of gradual alignment along with high growth rates of Tversky similarity with respect to $|X_u^{(i)}\cap Y_v^{(i)}|$ over iterations, which thereby results in improved alignment accuracy.

\begin{example}
We show how using the Tversky similarity ${\bf S}_{Tve}^{(i)}(u,v)$ is beneficial over the Jaccard index ${\bf S}_{Jac}^{(i)}(u,v)$ in terms of capturing the topological consistency of cross-network node pairs. As illustrated in Fig. \ref{fig:Example1}, consider two networks $G_s$ and $G_t$ with $n_s=6$ and $n_t=4$, where there is a single matched pair (B,b) such that $\mathcal{T}_A^{(1)}=\{b\}$, $X_A^{(1)}=\{b,C,D,E,F\}$, $Y_a^{(1)}=\{b,c,d\}$, and  $X_A^{(1)} \cap Y_a^{(1)} = \{b\}$. Suppose that $\alpha = \frac{1}{2}$ and $\beta = 1$. Then, we have ${\bf S}_{Tve}^{(1)}(A,a)=\frac{1}{5}$ and ${\bf S}_{Jac}^{(1)}(A,a)=\frac{1}{7}$. At the next gradual alignment step, we assume that (C,c) is newly aligned, which leads to $X_A^{(2)}=\{b,c,D,E,F\}, Y_a^{(2)}=\{b,c,d\},$ and  $X_A^{(2)} \cap Y_a^{(2)} = \{b,c\}$. Then, it follows that ${\bf S}_{Tve}^{(2)}(A,a)=\frac{4}{9}$ and ${\bf S}_{Jac}^{(2)}(A,a)=\frac{2}{7}$. Since the growth rate of the Tversky similarity and the Jaccard index with respect to the number of ACNs is given by $\frac{11}{45}$ and $\frac{1}{7}$, respectively, using the Tversky similarity is beneficial in terms of better capturing the structural consistency across different networks, which finally leads to performance improvement.
\end{example}
\begin{figure}
    \centering
    \includegraphics[scale=0.12]{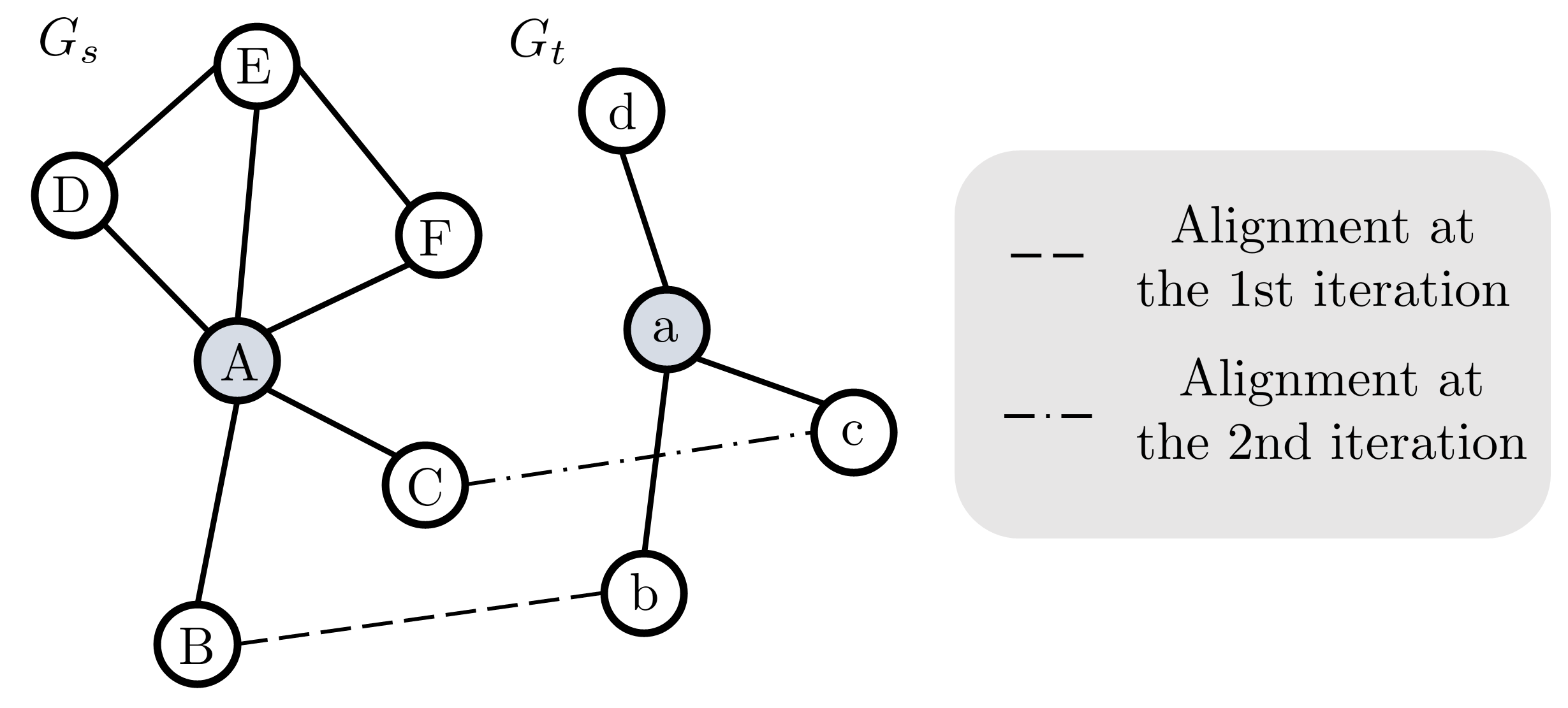}
    \caption{An example illustrating gradual alignment of two networks $G_s$ and $G_t$ with $n_s=6$ and $n_t=4$.}
    \label{fig:Example1}
\end{figure}    

After calculating our dual-perception similarity matrix ${\bf S}^{(i)}$ in (\ref{final sim}) based on two similarity matrices ${\bf S}_{emb}$ and ${\bf S}_{Tve}^{(i)}$, we gradually match $N$ node pairs for each iteration until all $M$ node pairs are found. Let us explain how to choose $N$ node pairs per iteration as follows. At the first iteration, we find the most confident pair (i.e., an element with the highest value) out of $n_s n_t$ elements in ${\bf S}^{(1)}$ and then find the second most confident pair among $(n_s-1)(n_t-1)$ elements after deleting all elements related to the most confident pair. We repeatedly perform the above steps until $N$ node pairs are matched for the first iteration. Updating the node mapping $\pi^{(i)}$, we recalculate ${\bf S}_{Tve}^{(i+1)}$ to discover $N$ node pairs at the next iteration. This process is repeated at each iteration. We refer to lines 18--22 in Algorithm \ref{mainalgorithm} for calculation of ${\bf S}^{(i+1)}$ upon the updated one-to-one node mapping $\pi^{(i+1)}$. 

\subsection{Complexity Analysis}
\label{sec 4.2}
In this subsection, we analyze the computational complexity of \textsf{Grad-Align}. When an untrained GNN model is used, the complexity of the multi-layer embedding similarity can be analyzed below. In the feed-forward process, the computational complexity of message passing over all GNN layers is given by $\mathcal{O}(L\max\{|\mathcal{E}_s|,|\mathcal{E}_t|\})$ \cite{wu2020comprehensive}, where $L$ is the number of GNN layers, and $\mathcal{E}_s$ and $\mathcal{E}_t$ are the numbers of edges in $G_s$ and $G_t$, respectively. While the element-wise calculation of the similarity matrix ${\bf S}_{emb}$ in (\ref{emb_sim}) is repeated $n_s n_t$ times \cite{du2019joint}, this process is computable in constant time when parallelization is applied. Now, we are ready to show the following theorem, which states a comprehensive analysis of the total complexity.
\begin{theorem}
\label{theorem_compelxity}
The computational complexity of the proposed \textsf{Grad-Align} method is given by $\mathcal{O}(\max\{|\mathcal{E}_s|,|\mathcal{E}_t|\})$.
\end{theorem}

\begin{proof}
For a given node pair $(u,v)$, we focus on analyzing the computational complexity of the Tversky similarity calculation. Since we take into account one-hop neighbors of each node in handling each term in (\ref{Tversky sim}), the complexity of calculating $|X_u^{(i)}-Y_v^{(i)}|$, $|Y_v^{(i)}-X_u^{(i)}|$, and $|X_u^{(i)} \cap Y_v^{(i)}|$ is given by $\mathcal{O}(|X_u^{(i)}|)$, $\mathcal{O}(|Y_v^{(i)}|)$, and $\mathcal{O}(|X_u^{(i)}||Y_v^{(i)}|)$, respectively, for the worst case, which are 
bounded by the maximum node degree in $G_s$ and $G_t$ and thus are regarded as a constant. As in the calculation of ${\bf S}_{emb}$, the Tversky similarity matrix ${\bf S}_{Tve}^{(i)}$ is computable in constant time due to the independent element-wise calculation via parallel processing. From the fact that the number of iterations for gradual alignment is finite and the computation of ${\bf S}_{emb}$ is a bottleneck,
the total complexity of our \textsf{Grad-Align} method is finally bounded by $\mathcal{O}(\max\{|\mathcal{E}_s|,|\mathcal{E}_t|\})$. This completes the proof of this theorem.
\end{proof}

From Theorem \ref{theorem_compelxity}, one can see that the computational complexity of \textsf{Grad-Align} scales linearly with the maximum number of edges over two networks. This also implies that the computational complexity of \textsf{Grad-Align} scales no larger than that of the embedding similarity calculation using the underlying GNN model.

\subsection{\textsf{Grad-Align-EA} Method}
\label{sec 4.3}
\begin{figure}
    \centering
    \includegraphics[scale=0.055]{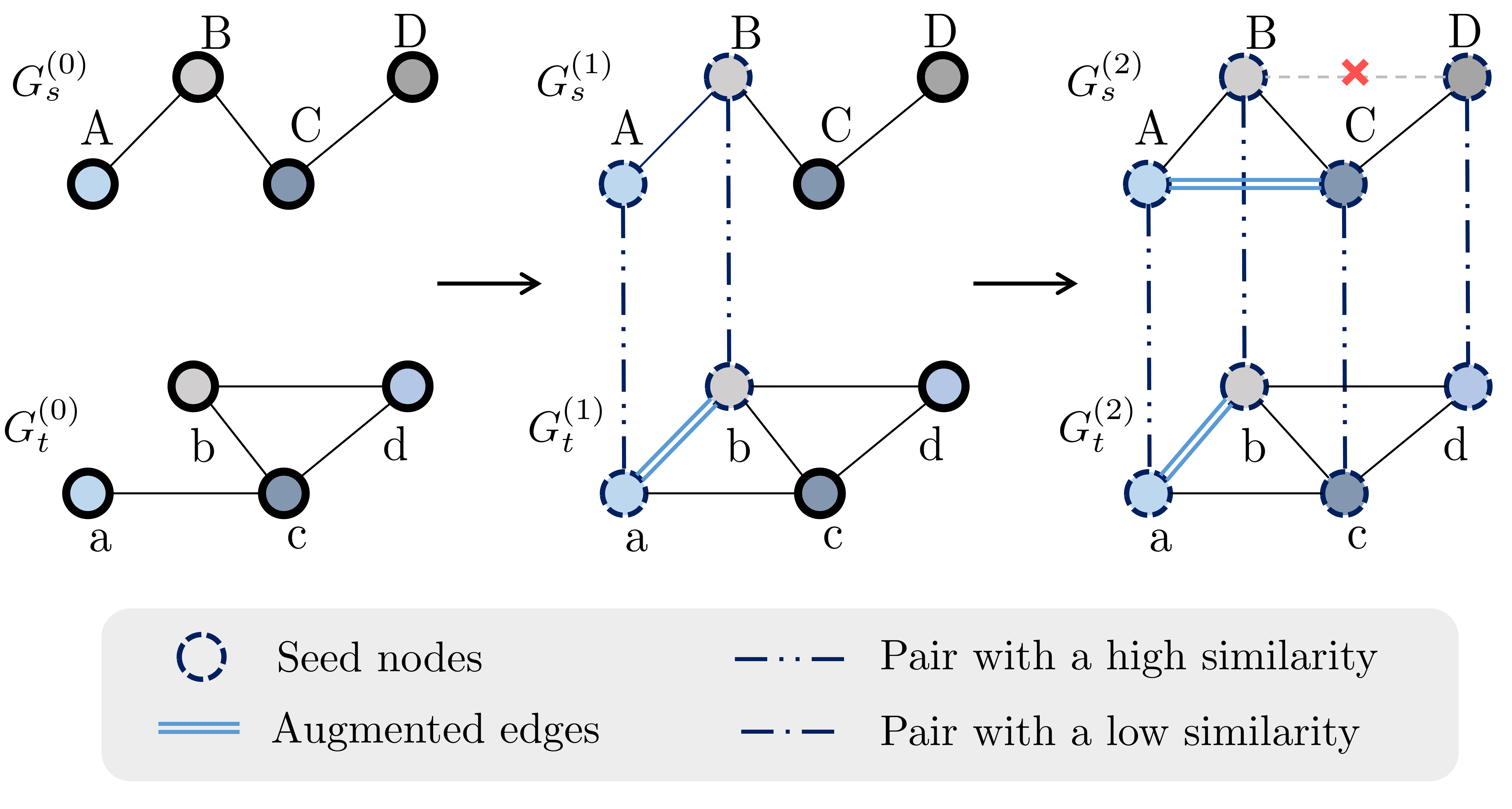}
    \caption{An example illustrating edge augmentation over iterations. Here, the edge between nodes $(a,b)$ in $G_t$ is augmented at the first step, and the edge between nodes $(A,C)$ is augmented at the second step.}
    \label{cea fig}
\end{figure}   
In this subsection, we introduce \textsf{Grad-Align-EA}, namely \textsf{Grad-Align} with edge augmentation, to further improve the performance of \textsf{Grad-Align} by incorporating an edge augmentation module into the \textsf{Grad-Align} method, which is inspired by CENALP \cite{du2019joint} that jointly performs NA and link prediction. However, different from the link prediction in CENALP that requires the high computational cost, our edge augmentation only creates edges that are {\em highly confident} without any training and test phases. The edge augmentation makes both networks $G_s$ and $G_t$ evolve in such a way that their structural consistencies across the networks are getting high over iterations in the gradual alignment process. To express the evolution of each network, we rewrite source and target networks as $G^{(i)}_s = (\mathcal{V}_s,\mathcal{E}^{(i)}_s,\mathcal{X}_s)$, $G_t^{(i)} = (\mathcal{V}_t, \mathcal{E}^{(i)}_t,\mathcal{X}_t)$, respectively, where $\mathcal{E}_s^{(i)}$ and $\mathcal{E}_t^{(i)}$ are the sets of edges in $G_s^{(i)}$ and $G_t^{(i)}$, respectively. The overall procedure of our edge augmentation is summarized in Algorithm \ref{CEAalgorithm}.\footnote{The source code for \text{Grad-Align-EA} is available online (https://github.com/jindeok/Grad-Align-full/tree/main/Grad-Align-EA).} Let $\tilde{\mathcal{E}}_s$ ($\tilde{\mathcal{E}}_t$) denote the set of {\em potentially confident edges} that are not in the set $\mathcal{E}_s^{(i)}$ ($\mathcal{E}_t^{(i)}$) while their counterpart edges upon the node mapping $\pi^{(i)}$ are in $\mathcal{E}_t^{(i)}$ ($\mathcal{E}_s^{(i)}$). If the aligned node pair does not match the ground truth, then the augmented edge will be treated as noise in the next alignment step. To remedy this problem, our edge augmentation module creates edges among potentially confident ones only when they are strongly confident. More specifically, we augment edges only when the dual-perception similarity ${\bf S}^{(i)}(u,v)$ of node pair $(u,v)$ is greater than a certain threshold $\tau>0$ (refer to lines 3 and 9 in Algorithm \ref{CEAalgorithm}), where $\tau$ will be set appropriately in Section \ref{sec5.5.5}. Fig. \ref{cea fig} illustrate an example of edge augmentation over iterations in our \textsf{Grad-Align-EA}. In the first step, the potentially confident edge between nodes $(a,b)$ in $G_t$ is augmented. In the second step, the edge between nodes $(A,C)$ is augmented but another potentially confident edge between nodes $(B,D)$ is not augmented because ${\bf S}^{(i)}(D,d)$ is not sufficiently high.

Since the structure of each network evolves over iterations via the edge augmentation module, we need to newly update the hidden representations of each network by retraining model parameters $\theta$ of GNN and then recalculate the multi-layer embedding similarity matrix. In this context, the dual-perception similarity matrix in (\ref{final sim}) is rewritten as
    \begin{equation}
    \label{CEA_dual sim}
    \mathbf{S}^{(i)}=\mathbf{S}^{(i)}_{emb}\odot \mathbf{S}^{(i)}_{Tve},
    \end{equation}
where $\mathbf{S}^{(i)}_{emb}$ represents the multi-layer embedding similarity matrix at the $i$-th iteration when \textsf{Grad-Align-EA} is used. 

\begin{algorithm}[t]
\caption{Edge augmentation}
\label{CEAalgorithm}
 \begin{algorithmic}[1]
  \renewcommand{\algorithmicrequire}{\textbf{Input:}}
  \renewcommand{\algorithmicensure}{\textbf{Output:}}
  \REQUIRE $G^{(i)}_s = (\mathcal{V}_s,\mathcal{E}^{(i)}_s,\mathcal{X}_s)$, $G_t^{(i)} = (\mathcal{V}_t, \mathcal{E}^{(i)}_t,\mathcal{X}_t)$, $\pi^{(i)}$, $\mathbf{S}^{(i)}$, $\tilde{\mathcal{E}}_s$, $\tilde{\mathcal{E}}_t$, $\tau$
  \ENSURE $G^{(i+1)}_s, G^{(i+1)}_t$
  \STATE /* Edge augmentation for $G^{(i)}_s$ */  
  \FOR{$(u_1, u_2) \in \tilde{\mathcal{E}}_s$}
  \IF{$\mathbf{S}^{(i)}(u,\pi^{(i)}{(u)}) > \tau$ and $\mathbf{S}^{(i)}(v,\pi^{(i)}{(v)}) > \tau$}
  \STATE $\mathcal{E}^{(i+1)}_s = \mathcal{E}^{(i)}_s \cup \{(u, v)\}$
  \ENDIF
  \ENDFOR
  \STATE /* Edge augmentation for $G^{(i)}_t$ */ 
  \FOR{$(v_1, v_2) \in \tilde{\mathcal{E}}_t$}
  \IF{$\mathbf{S}^{(i)}(\pi^{-1(i)}{(u'),u'}) > \tau$ and $\mathbf{S}^{(i)}(\pi^{-1(i)}{(v'),v'}) > \tau$}
  \STATE $\mathcal{E}^{(i+1)}_t = \mathcal{E}^{(i)}_t \cup \{ (u',v')\}$
  \ENDIF
  \ENDFOR
  
  \RETURN $G^{(i+1)}_s=(\mathcal{V}_s,\mathcal{E}^{(i+1)}_s,\mathcal{X}_s)$,\\ $G^{(i+1)}_t = (\mathcal{V}_t, \mathcal{E}^{(i+1)}_t,\mathcal{X}_t)$
 \end{algorithmic}
\end{algorithm}


\section{Experimental Evaluation}\label{section 5}
In this section, we first describe datasets used in the evaluation. We also present five state-of-the-art NA methods for comparison. After presenting performance metrics and our experimental settings, we comprehensively evaluate the performance of our \textsf{Grad-Align} method and five benchmark methods.

\subsection{Datasets}

\begin{table}[t!]
\scriptsize
\begin{tabular}{cccccc}
\hline
\multicolumn{2}{c}{Datasets}                                                 & {\begin{tabular}[c]{@{}c@{}} \# of \\ nodes\end{tabular}} & {\begin{tabular}[c]{@{}c@{}} \# of \\ edges\end{tabular}} & {\begin{tabular}[c]{@{}c@{}} \# of \\ attributes\end{tabular}} & {\begin{tabular}[c]{@{}c@{}} \# of \\ ground truth\\node pairs\end{tabular}}          \\ \hline
\multirow{2}{*}{\begin{tabular}[c]{@{}c@{}}Facebook \\Twitter\end{tabular}}    & $G_s$  & 1,043 & 4,734 & 0  & \multirow{2}{*}{1,043} \\
                                                                        & $G_t$ & 1,043 & 4,860  & 0 &                       \\ \hline
\multirow{2}{*}{\begin{tabular}[c]{@{}c@{}}Douban Online \\ Douban Offline\end{tabular}} & $G_s$  & 3,906 & 8,164 & 538  & \multirow{2}{*}{1,118} \\
                                                                        & $G_t$ & 1,118 & 1,511 & 538  &                       \\ \hline
\multirow{2}{*}{\begin{tabular}[c]{@{}c@{}}Allmovie \\ IMDb\end{tabular}}   & $G_s$  & 6,011  & 124,709 & 14  & \multirow{2}{*}{5,176} \\
                                                                        & $G_t$ & 5,713 & 119,073 & 14  &                       \\ \hline
\multirow{2}{*}{\begin{tabular}[c]{@{}c@{}}Facebook\\(Its noisy version)\end{tabular}}    & $G_s$  & 1,256 & 4,260 & 0  & \multirow{2}{*}{1,043} \\
                                                                        & $G_t$ & 1,256 & 4,256  & 0 &                       \\ \hline
\multirow{2}{*}{\begin{tabular}[c]{@{}c@{}}Econ\\(Its noisy version)\end{tabular}} & $G_s$  & 1,258 & 6,857 & 20  & \multirow{2}{*}{1,258} \\
                                                                        & $G_t$ & 1,258 & 6,860 & 20  &                       \\ \hline
\multirow{2}{*}{\begin{tabular}[c]{@{}c@{}}DBLP\\(Its noisy version)\end{tabular}}                                                   & $G_s$  & 2,151  & 5,676 & 20  & \multirow{2}{*}{2,158} \\
                                                                        & $G_t$ & 2,151 & 5,672 & 20  &                       \\ \hline
\multirow{2}{*}{\begin{tabular}[c]{@{}c@{}}Foursquare\\(Its noisy version)\end{tabular}}                                                   & $G_s$  & 17,355  & 132,208 & 0  & \multirow{2}{*}{17,355} \\
                                                                        & $G_t$ & 17,355 & 131,018 & 0  &                       \\ \hline
\end{tabular}
\caption{Statistics of the seven datasets used in our experiments.}
\label{datasettable}
\end{table}

We conduct experiments on several real-world and synthetic datasets across various domains, which are widely adopted for evaluating the performance of NA. The main statistics of each dataset, including the number of nodes, the number of edges, the number of attributes, and the number of node pairs, are summarized in Table \ref{datasettable}. In the following, we explain important characteristics of the datasets briefly.
\subsubsection{Real-World Datasets}
\label{sec 5.1.1}
We use three real-world datasets, each of which consists of two networks (i.e., source and target networks).

\textbf{Facebook vs. Twitter} {\bf (Fb-Tw)}. The Fb-Tw dataset is composed of two real-world social networks collected and published by \cite{cao2016bass}. User accounts and friendships of accounts are treated as nodes and edges, respectively. 

\textbf{Douban Online vs. Douban Offline} {\bf (Douban)}. The Douban dataset is a Chinese social network collected and published by \cite{zhong2012comsoc}. User accounts and their friendships are treated as nodes and edges, respectively. 

\textbf{Allmovie vs. IMDb} {\bf (Am-ID)}. The Allmovie network is constructed from Rotten Tomatoes (an American review-aggregation website), where films are treated as nodes and two films have an edge connecting them if they have at least one common actor.\footnote{https://www.kaggle.com/ayushkalla1/rotten-tomatoes-movie-database.} The IMDb network is constructed from IMDb (an online database of information on movies, TV, and celebrities).\footnote{https://www.kaggle.com/jyoti1706/IMDBmoviesdataset.} Its nodes and edges are created similarly as in Allmovie. 

\subsubsection{Synthetic Datasets}
\label{sec 5.1.2}
In addition to the above three real-world datasets, we synthesize network data to comprehensively evaluate noisy conditions on the network structure and node attributes similarly as in \cite{trung2020adaptive,du2019joint}. From the original network, we generate its noisy version by randomly removing a certain number of edges and replacing a portion of node attributes with zeros, while preserving the number of ground truth cross-network node pairs. We use four synthetic datasets, consisting of two non-attributed networks and two attributed networks, as follows.

\textbf{Facebook}. The Facebook network is originated from the source network of the Fb-Tw dataset in Section \ref{sec 5.1.1}.

\textbf{Econ}. The Econ network is an economic model of Victoria state, Australia during the banking crisis in 1880 \cite{rossi2015network, trung2020adaptive}. The nodes and edges represent the organizations located in the state and the contractual relationships between them, respectively.

\textbf{DBLP}. The DBLP dataset is a co-authorship network collected and published by \cite{prado2012mining}. Authors and their academic
interactions are treated as nodes and edges, respectively. An attribute vector represents the number of publications in computer science conferences \cite{du2019joint}.

\textbf{Foursquare}. The Foursquare dataset is a location-based online social network and is originally collected by \cite{zhang2015integrated}. The nodes and edges represent the users and the follower/followee relationships between them, respectively.

\subsection{State-of-the-Art Methods}
In this subsection, we present five state-of-the-art NA methods for comparison.

\textbf{GAlign}~\cite{trung2020adaptive}. This is an unsupervised NA method based on a multi-order GCN model using local and global structural information, where prior seed nodes are not available.

\textbf{CENALP}~\cite{du2019joint}. This method jointly performs NA and link prediction tasks to improve the alignment accuracy. A cross-network embedding strategy is employed based on a variant of DeepWalk \cite{perozzi2014deepwalk}.

\textbf{FINAL}~\cite{zhang2016final}. This method aligns attributed networks based on the alignment consistency principle. Specifically, FINAL utilizes three consistency conditions including the topology consistency, node attribute consistency, and edge attribute consistency.

\textbf{PALE}~\cite{man2016predict}. This is a supervised NA model that employs network embedding with awareness of prior seed nodes and learns a cross-network mapping via an MLP architecture for NA.

\textbf{DeepLink}~\cite{zhou2018deeplink}. This model encodes nodes into vector representations to capture local and global network structures through deep neural networks in a semi-supervised learning manner.

\subsection{Performance Metrics}
To assess the performance of our proposed \textsf{Grad-Align} method and the five state-of-the-art methods, as the most popular metric, we adopt the {\em alignment accuracy} \cite{zhang2016final, du2019joint}, denoted as {\em Acc}, which quantifies the proportion of correct node correspondences out of the total $M$ correspondences. We also adopt {\em Precision@q} (also known as $Success@q$) as another performance metric as in \cite{zhou2018deeplink, trung2020adaptive, zhang2016final}, which indicates whether there is the true positive matching identity in top-$q$ candidates and is expressed as
\begin{equation}
    Precision@q = \frac{\sum_{v_s^*\in\mathcal{V}_s} \mathds{1}_{v_t^* \in \mathcal{S}^q(v_s^*)}}{M},
\end{equation}
where $(v_s^*, v_t^*)$ is each node pair in the ground truth; $\mathcal{S}^q(v_s^*)$ indicates the set of indices of top-$q$ elements in the $v_s^*$-th row of the dual-perception similarity matrix ${\bf S}^{(\ceil*{\frac{M}{N}}+1)}$ in (\ref{final sim}); and $\mathds{1}_{\mathcal{S}_{v_s^*}^q(v_t^*)}$ is the indicator function. For node $v_s^*$, if the similarity ${\bf S}^{(\ceil*{\frac{M}{N}}+1)}(v_s^*, v_t^*)$ is ranked within the $q$-th highest values in the row ${\bf S}^{(\ceil*{\frac{M}{N}}+1)}(v_s^*,:)$ of the similarity matrix ${\bf S}^{(\ceil*{\frac{M}{N}}+1)}$, then the alignment output for $v_s^*$ is recorded as a successful case. Note that the higher the value of each of the two metrics, the better the performance.

\subsection{Experimental Setup}
\label{sec 5.4}
 We describe experimental settings of neural networks (i.e., the GNN model) in our \textsf{Grad-Align} method. The GNN model is implemented by PyTorch Geometric \cite{fey2019pyg}, which is a geometric deep learning extension library in PyTorch. We set the dimension of each GNN hidden layer as $150$. We train our GNN model using Adam optimizer \cite{kingma2015adam} with a learning rate of 0.005. Since Fb-Tw, Facebook (and its noisy version), and Foursquare (and its noisy version) datasets do not contain node attributes, we use all-ones vectors $\mathbf{1} \in \mathbb{R}^{1 \times n_s}$ and $\mathbf{1} \in \mathbb{R}^{1 \times n_t}$ as the input node attribute vectors for the GNN model on the datasets. We use the following key hyperparameters in \textsf{Grad-Align}.
\begin{itemize}
\item The number of GNN layers ($k$);
\item Coefficients of the Tversky index in (\ref{Tversky sim}) ($\alpha$ and $\beta$);
\item The number of iterations for gradual matching ($iter=\ceil{\frac{M}{N}} +1$);
\item The proportion of prior seed node pairs out of $M$ ground truth node pairs ($t$).
\end{itemize}
In our experiments, the above hyperparameters are set to the pivot values $k=2$, $\alpha=\frac{n_s}{n_t}$, $\beta=1$, $iter=15$, and $t=0.1$ unless otherwise stated.
 
 Next, as a default experimental setting for all the methods, 10\% of ground truth node correspondences are randomly selected and used as the training set over all the datasets as in \cite{trung2020adaptive}. For each synthetic dataset in Section \ref{sec 5.1.2}, to generate its noisy version, we randomly remove 10\% of edges and replace 10\% of node attributes with zeros unless otherwise specified. We conduct each experiment over 10 different random seeds to evaluate the average performance. All experiments are carried out with Intel (R) 12-Core (TM) i7-9700K CPUs @ 3.60 GHz and 32GB RAM.

 \subsection{Experimental Results}
 
 In this subsection, our extensive empirical study is designed to answer the following six key research questions.
\begin{itemize}
    \item \textit{Q1.} How do underlying GNN models affect the performance of the \textsf{Grad-Align} method?
    \item \textit{Q2.} How do model hyperparameters affect the performance of the \textsf{Grad-Align} method?
    \item \textit{Q3.} How much does the \textsf{Grad-Align} method improve the NA performance over state-of-the-art NA methods?
    \item \textit{Q4.} How robust is our \textsf{Grad-Align} method to more difficult settings with high structural/attribute noise levels?
    \item \textit{Q5.} How much does each component in the \textsf{Grad-Align} method contributes to the performance?
    \item \textit{Q6.} How expensive is the computational complexity of \textsf{Grad-Align} in comparison with state-of-the-art NA methods?
\end{itemize}
To answer the research questions stated above, we comprehensively carry out experiments in the following.

\subsubsection{Comparative Study Among GNN Models (Q1)}
\begin{figure}[t!]
    \centering
    \pgfplotsset{compat=1.11,
    /pgfplots/ybar legend/.style={
    /pgfplots/legend image code/.code={%
       \draw[##1,/tikz/.cd,yshift=-0.25em]
        (0cm,0cm) rectangle (3pt,0.8em);},
   },
}
    \begin{tikzpicture}
    \begin{axis}[
        width  = 0.65*\columnwidth,
        height = 3.5cm,
        major x tick style = transparent,
        ybar=0,
        bar width=0.03*\columnwidth,
        ymajorgrids = true,
        ylabel = {{\em Acc}},
        symbolic x coords={Fb-Tw, Douban, Am-ID},
        xtick = data,
        scaled y ticks = false,
        enlarge x limits=0.4,
        ymin=0.4,
        legend cell align=left,
        legend style={at={(0.5,1.1)}, anchor=south,legend columns=3,font=\footnotesize}
    ]
        \addplot[style={black,fill=teal,mark=none}, error bars/.cd,
y dir=both,y explicit]
            coordinates {(Fb-Tw, 0.9674) +-(0.0127,0.0187) (Douban,0.5760)+-(0.0207,0.0187) (Am-ID,0.9508)+-(0.0227,0.0185)};

        \addplot[style={black,fill=lightgray,mark=none}, error bars/.cd,
y dir=both,y explicit]
            coordinates {(Fb-Tw, 0.9594) +-(0.0127,0.0187) (Douban,0.5700)+-(0.0207,0.0187) (Am-ID,0.9458)+-(0.0227,0.0185)};

        \addplot[style={black,fill=brown,mark=none}, error bars/.cd,
y dir=both,y explicit]
            coordinates {(Fb-Tw, 0.9674) +-(0.0127,0.0187) (Douban,0.5921)+-(0.0207,0.0187) (Am-ID,0.9618)+-(0.0227,0.0185)};

        \legend{\textsf{Grad-Align-GCN},\textsf{Grad-Align-SAGE},\textsf{Grad-Align-GIN}}
    \end{axis}
\end{tikzpicture}
    \caption{Alignment accuracy according to different GNN models in our \textsf{Grad-Align} method on the three real-world datasets.}
    \label{Q1plot_gnn}
\end{figure}
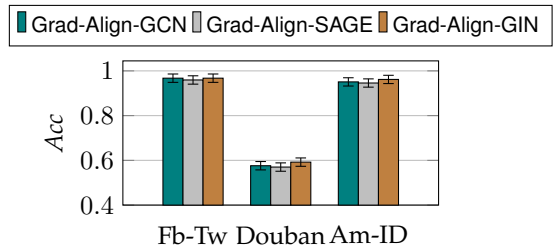

\begin{figure}[t!]
    \centering
    \pgfplotsset{compat=1.11,
    /pgfplots/ybar legend/.style={
    /pgfplots/legend image code/.code={%
       \draw[##1,/tikz/.cd,yshift=-0.25em]
        (0cm,0cm) rectangle (3pt,0.8em);},
   },
}
\begin{tikzpicture}
    \begin{axis}[
        width  = 0.65*\columnwidth,
        height = 3.5cm,
        major x tick style = transparent,
        ybar=0,
        bar width=0.03*\columnwidth,
        ymajorgrids = true,
        ylabel = {{\em Acc}},
        symbolic x coords={Fb-Tw, Douban, Am-ID},
        xtick = data,
        scaled y ticks = false,
        enlarge x limits=0.4,
        ymin=0.4,
        legend cell align=left,
        legend style={at={(0.5,1.1)}, anchor=south,legend columns=3,font=\footnotesize}
    ]
        \addplot[style={black,fill=teal,mark=none}, error bars/.cd,
y dir=both,y explicit]
            coordinates {(Fb-Tw, 0.9674) +-(0.0127,0.0187) (Douban,0.5921)+-(0.0207,0.0187) (Am-ID,0.9618)+-(0.0227,0.0185)};

        \addplot[style={black,fill=lightgray,mark=none}, error bars/.cd,
y dir=both,y explicit]
            coordinates {(Fb-Tw, 0.9501) +-(0.0127,0.0187) (Douban,0.5823)+-(0.0207,0.0187) (Am-ID,0.9518)+-(0.0227,0.0185)};

        \addplot[style={black,fill=brown,mark=none}, error bars/.cd,
y dir=both,y explicit]
            coordinates {(Fb-Tw, 0.9559) +-(0.0127,0.0187) (Douban,0.5801)+-(0.0207,0.0187) (Am-ID,0.9501)+-(0.0227,0.0185)};

        \legend{\textsf{Sum},\textsf{Mean},\textsf{Max}}

    \end{axis}
\end{tikzpicture}
    \caption{Alignment accuracy according to different aggregators in our \textsf{Grad-Align} method on the three real-world datasets.}--
    \label{Q1plot_agg}
\end{figure}
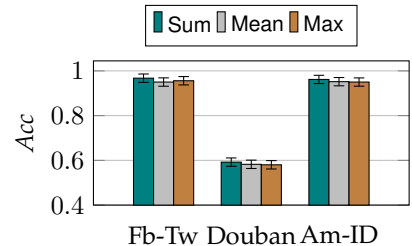
In Fig. \ref{Q1plot_gnn}, we show the performance on {\em Acc} of various GNN models in \textsf{Grad-Align} using three real-world datasets. Since our method is {\em GNN-model-agnostic}, any existing models can be adopted; however, in our experiments, we adopt the following three milestone GNN models from the literature, namely GCN \cite{DBLP:conf/iclr/KipfW17} (\textsf{Grad-Align-GCN}), GraphSAGE \cite{DBLP:conf/nips/HamiltonYL17} (\textsf{Grad-Align-SAGE}), and GIN \cite{xu2018powerful} (\textsf{Grad-Align-GIN}).\footnote{For the GIN model, we use a two-layer MLP architecture and the ReLU activation function~\cite{DBLP:conf/icml/NairH10} according to the original implementation in~\cite{xu2018powerful}.} Here, we use the sum aggregator as a default AGGREGATE function in (\ref{aggregateequation}).

From Fig. \ref{Q1plot_gnn}, it is seen that \textsf{Grad-Align-GIN} consistently outperforms other models regardless of the datasets although the gains over other models are not significant. It is worth noting that GIN was proposed to generalize the Weisfeiler-Lehman (WL) graph isomorphism test to achieve its maximum discriminative power among GNNs \cite{xu2018powerful}. Due to the fact that the higher the expressiveness of the node representation, the more probable it is to find the correct node correspondence without any ambiguity, the GIN model-aided approach (i.e., \textsf{Grad-Align-GIN}) achieves the best performance with a more expressive power.

In Fig. \ref{Q1plot_agg}, we show how {\em Acc} performance behaves according to several AGGREGATE functions in (\ref{aggregateequation}), including the sum, mean, and max aggregators, when \textsf{Grad-Align-GIN} is used for the three real-world datasets. It is observed that the sum aggregator has the most expressive power among three AGGREGATE functions while achieving the highest {\em Acc}. This finding is consistent with the statements in \cite{xu2018powerful}.

From the intriguing observations stated above, we use \textsf{Grad-Align-GIN} with sum aggregation in our subsequent experiments.

\subsubsection{Effect of Hyperparameters (Q2)}
\label{sec5.5.2}

In Fig. \ref{Q2plot}, we investigate the impact of four key hyperparameters, including $k$, $\alpha$, $iter$, and $t$ addressed in Section \ref{sec 5.4}, on the performance of \textsf{Grad-Align} in terms of the {\em Acc} score using three real-world datasets. When a hyperparameter varies so that its effect is clearly revealed, other parameters are set to the pivot values in Section \ref{sec 5.4}. Our findings are as follows.
\begin{figure}[t!]
\pgfplotsset{footnotesize,samples=10}
\centering
\begin{tikzpicture}
\begin{axis}[
legend columns=-1,
legend entries={Fb-Tw, Douban, Am-ID},
legend to name=named,
xmax=4,xmin=1,ymin= 0.5,ymax=1,
xlabel=(a) Effect of $k$,ylabel={\em Acc}, width = 4cm, height = 3.6cm,
xtick={1,2,3,4},ytick={0.5,0.6,0.7,...,1}]
    \addplot+[color=black] coordinates{(1,0.9602) (2,0.9678) (3,0.9553)(4,0.9542)};
    \addplot+[color=orange] coordinates{(1,0.5428) (2,0.5925) (3,0.5687)(4,0.5524)};
    \addplot+[color=purple] coordinates{(1,0.9511) (2,0.9577) (3,0.9497)(4,0.9552)};
\end{axis}
\end{tikzpicture}
\begin{tikzpicture}
\begin{axis}[
xmax=1,xmin=0.1,ymin= 0.5,ymax=1,
xlabel=(b) Effect of $\alpha$,ylabel={\em Acc}, width = 4cm, height = 3.6cm,
xtick={0.1,0.3,0.5,0.7,0.9,1},ytick={0.5,0.6,0.7,...,1}]
    \addplot+[color=black] coordinates{(0.1,0.9559) (0.3,0.9505) (0.5,0.9602) (0.7,0.9608) (0.9,0.9651) (1,0.9678)};
    \addplot+[color=orange] coordinates{(0.1,0.5312) (0.3,0.5921) (0.5,0.5631) (0.7,0.5679) (0.9,0.5633) (1,0.5323)};
    \addplot+[color=purple] coordinates{(0.1,0.9545) (0.3,0.9512) (0.5,0.9532) (0.7,0.9502) (0.9,0.9577) (1,0.9581)};
\end{axis}
\end{tikzpicture}
\begin{tikzpicture}
\begin{axis}[
xmax=20,xmin=1,ymin= 0.2,ymax=1,
xlabel=(c) Effect of $iter$,ylabel={\em Acc}, width = 4cm, height = 3.6cm,
xtick={1,5,10,15,20},ytick={0,0.2,0.4,...,1}]
    \addplot+[color=black]  coordinates{(1,0.4669)(5,0.9128)(10,0.9540)(15,0.9674)(20,0.9697)};
    \addplot+[color=orange] coordinates{(1,0.3241)(5,0.5286)(10,0.5581)(15,0.5921)(20,0.5931)};
    \addplot+[color=purple] coordinates{(1,0.8942)(5,0.9514)(10,0.9558)(15,0.9568)(20,0.9574)};
\end{axis}
\end{tikzpicture}
\begin{tikzpicture}
\begin{axis}[
xmax=0.5,xmin=0,ymin= 0,ymax=1,
xlabel=(d) Effect of $t$,ylabel={\em Acc}, width = 4cm, height = 3.6cm,
xtick={0,0.1,0.2,0.3,0.4,0.5},ytick={0,0.2,0.4,...,1}]
    \addplot+[color=black] coordinates{(0,0.05) (0.1,0.9678) (0.2,0.9878) (0.3,0.9899) (0.4,0.9925) (0.5, 0.9935)};
    \addplot+[color=orange] coordinates{(0,0.3408) (0.1,0.5921) (0.2,0.6879) (0.3,0.7487) (0.4,0.7907) (0.5, 0.8504)};
    \addplot+[color=purple] coordinates{(0,0.8758) (0.1,0.9577) (0.2,0.9711) (0.3,0.9840) (0.4,0.9942) (0.5, 0.9978)};
\end{axis}
\end{tikzpicture}
\\
\ref{named}
\caption{Alignment accuracy according to different values of hyperparameters in our \textsf{Grad-Align} method on the three real-world datasets.}
\label{Q2plot}
\end{figure}
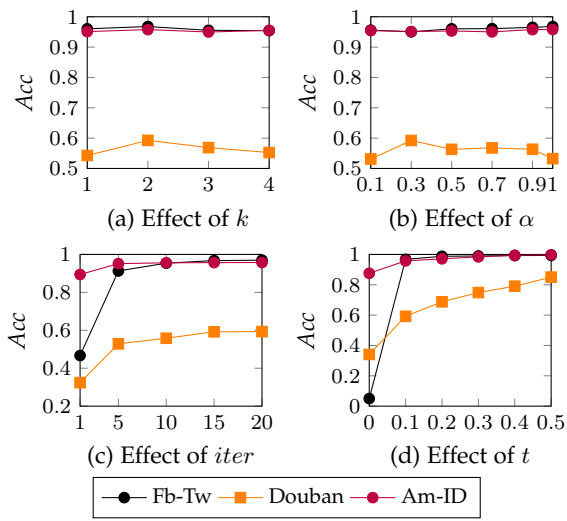
%
\begin{table*}[h!]
\centering
\setlength\tabcolsep{8pt} 
\scalebox{1}{
\begin{tabular}{cccccccc}
\toprule 
Dataset & Metric &\textsf{Grad-Align} & GAlign & CENALP & FINAL & PALE & DeepLink\\
\midrule
\multirow{4}{*}{\rotatebox{0}{Fb-Tw}} 
  & {\em Acc} & \underline{0.9674}  & 0.0513 & 0.8566 & 0.8571 & 0.7963  & 0.2520\\
  & $Precision@1$ & \underline{0.9674} & 0.0306 & 0.5781 & 0.6007 & 0.5295  & 0.1151\\
  & $Precision@5$ & \underline{0.9856}  & 0.0422 & 0.7220 & 0.7864 & 0.7854  & 0.2927\\
  & $Precision@10$ & \underline{0.9904} & 0.0612& 0.7661 & 0.7773 & 0.8803  & 0.3350\\
\midrule
\multirow{4}{*}{\rotatebox{0}{Douban}} 
  & {\em Acc} & \underline{0.5921} & 0.2568 & 0.1350 & 0.3373 & 0.1852  & 0.1011\\
  & $Precision@1$ & \underline{0.6682}  & 0.3686 & 0.1571 & 0.3479 & 0.1377  & 0.0921\\
  & $Precision@5$ & \underline{0.8551}  & 0.5233 & 0.2584 & 0.5134 & 0.3283  & 0.2290\\
  & $Precision@10$ & \underline{0.8980} & 0.6324 & 0.3618 & 0.6764 & 0.4338  & 0.2907\\
\midrule
\multirow{4}{*}{\rotatebox{0}{Am-ID}} 
  & {\em Acc} & \underline{0.9568}  & 0.7364 & 0.5212 & 0.7725 & 0.7397  & 0.2780\\
  & $Precision@1$ & \underline{0.9601} & 0.7251 & 0.4566& 0.6524 & 0.5793  & 0.1361\\
  & $Precision@5$ & \underline{0.9744} & 0.8101 & 0.6657 & 0.8594 & 0.8242  & 0.3327\\
  & $Precision@10$ & \underline{0.9783}  & 0.8749 & 0.8127 & 0.8945 & 0.8519  & 0.4650\\
\midrule
\multirow{4}{*}{\rotatebox{0}{Facebook \& its noisy version}} 

  & {\em Acc} & \underline{0.9396} & 0.0413 & 0.8174& 0.6072 & 0.6079  & 0.2780\\
  & $Precision@1$ & \underline{0.9415}  & 0.0406 & 0.8174 & 0.5004 & 0.5279  & 0.1361\\
  & $Precision@5$ & \underline{0.9962}  & 0.0622 & 0.8523 & 0.5509 & 0.6040  & 0.3327\\
  & $Precision@10$ & \underline{0.9987}  & 0.0899 & 0.9175 & 0.8954 & 0.7162  & 0.4650\\
\midrule
\multirow{4}{*}{\rotatebox{0}{Econ \& its noisy version}} 
  & {\em Acc} & \underline{0.9841}  & 0.8953 & 0.5509 & 0.4948 & 0.5819  & 0.2583\\
  & $Precision@1$ & \underline{0.9873}  & 0.8647 & 0.4247 & 0.4434 & 0.3060  & 0.1161\\
  & $Precision@5$ & \underline{0.9960}  & 0.9100 & 0.5658 & 0.5642 & 0.6041  & 0.2671\\
  & $Precision@10$ & \underline{0.9968}  & 0.9346 & 0.7112 & 0.5944 & 0.6836  & 0.3506\\
\midrule
\multirow{4}{*}{\rotatebox{0}{DBLP \& its noisy version}} 
  & {\em Acc} & \underline{0.9350}  & 0.9126 & 0.6458 & 0.7646 & 0.3436  & 0.1413\\
  & $Precision@1$ & \underline{0.9292} & 0.8914 & 0.5992 & 0.7785 & 0.1437  & 0.0656\\
  & $Precision@5$ & \underline{0.9778} & 0.9340 & 0.6528 & 0.7954 & 0.3622  & 0.1957\\
  & $Precision@10$ & \underline{0.9990} & 0.9623 & 0.7775 & 0.8164 & 0.4826  & 0.2748\\
\midrule
\multirow{4}{*}{\rotatebox{0}{Foursquare \& its noisy version}} 
  & {\em Acc} & \underline{0.8998}  & 0.0241 & 0.8128 & 0.5214 & 0.3872  & 0.1374\\
  & $Precision@1$ & \underline{0.8992} & 0.0214 & 0.8112 & 0.5074 & 0.3778  & 0.1352\\
  & $Precision@5$ & \underline{0.9128} & 0.0287 & 0.8424 & 0.5876 & 0.3958  & 0.1627\\
  & $Precision@10$ & \underline{0.9344} & 0.0397 & 0.8875 & 0.6124 & 0.4124  & 0.1871\\
\bottomrule
\end{tabular}}
\caption{Performance comparison among \textsf{Grad-Align} and state-of-the-art NA methods in terms of the {\em Acc} and {\em Precision@q}. Here, the best method for each case is highlighted using underlines.}
\label{Q3table}
\end{table*}
\begin{itemize}
    \item {\bf The effect of $k$:} From Fig. \ref{Q2plot}a, setting $k=2$ consistently leads to the best performance for all the datasets. If $k<2$, then the proximity information is limited to the direct neighbors, which can be vulnerable to structural and attribute inconsistencies of nodes. If $k>2$, then our \textsf{Grad-Align} method may experience the oversmoothing problem for GNNs \cite{chen2020measuring}.

    \item {\bf The effect of $\alpha$:} From Theorem \ref{theorem_tversky}, we recall that using the Tversky similarity is beneficial in conducting the NA task compared to the case of the Jaccard similarity (i.e., $\alpha=\beta=1$). From Fig. \ref{Q2plot}b, we show that setting $\alpha$ near $\frac{n_t}{n_s}$ achieves the best performance for all the datasets, where $\frac{n_t}{n_s}$ corresponds to 1, 0.29, and 0.95 for the Fb-Tw, Douban, and Am-ID datasets, respectively. This empirical finding is vital in the design perspective. Specifically, by setting $\alpha\simeq \frac{n_t}{n_s}$ according to the ratio of the number of nodes in two unbalanced networks, we are capable of balancing between two terms in (\ref{Tversky sim}), $|X_u^{(i)}-Y_v^{(i)}|$ and $|Y_v^{(i)}-X_u^{(i)}|$, which thus enhances the degree of topological consistency across the two networks and then improves the alignment accuracy. This demonstrates the effectiveness of our Tversky similarity with proper parameter settings.
    
    \item {\bf The effect of $iter$:} In Fig. \ref{Q2plot}c, the performance is improved with increasing {\em iter} regardless of the datasets, which verifies the power of gradual alignment. It is also worthwhile to note that using only a small number of iterations for some datasets (e.g., $iter=5$ for Am-ID) is sufficient to achieve a significant gain over the case of finding all node correspondences at once (i.e., $iter=1$).

    \item {\bf The effect of $t$:}  As shown in Fig. \ref{Q2plot}d, one can see that dramatic gains are possible by exploiting prior seed nodes. For the non-attributed network such as Fb-Tw, the case of $t=0$ performs quite poorly. This is because it is difficult to precisely compute the multi-layer embedding similarity ${\bf S}_{emb}$ when node attributes are unavailable; thus, the Tversky similarity calculation with prior seed nodes plays a crucial role in correctly finding node correspondences for such non-attributed networks. On the other hand, for attributed networks such as Douban and Am-ID, satisfactory performance is observed even when $t=0$. 
\end{itemize}   

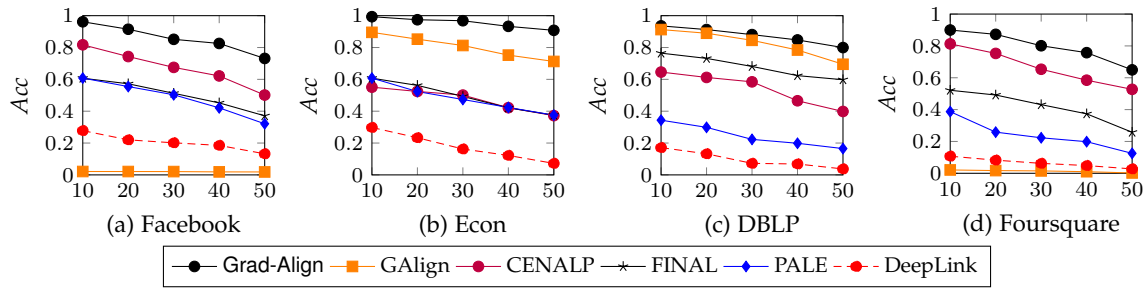
\begin{figure*}[t!]
\pgfplotsset{footnotesize,samples=10}
\centering
\begin{tikzpicture}
\begin{axis}[
legend columns=6,
legend entries={\textsf{Grad-Align},GAlign,CENALP,FINAL,PALE,DeepLink},
legend to name=named,
xlabel= (a) Facebook,ylabel={\em Acc},  width = 4cm, height = 3.7cm,
xmin=10,xmax=50,ymin= 0,ymax=1,
xtick={10,20,...,50},ytick={0,0.2,0.4,...,1}]
    \addplot+[color=black] coordinates{(10,0.9616)(20,0.9148)(30,0.8514)(40,0.8255)(50,0.7314)};
    \addplot+[color=orange]     coordinates{(10,0.0213)(20,0.0211)(30,0.0208)(40,0.0189)(50,0.0185)};
    \addplot+[color=purple] coordinates{(10,0.8170)(20,0.7432)(30,0.6752)(40,0.6215)(50,0.5014)};
    \addplot coordinates{(10,0.6072)(20,0.5721)(30,0.5125)(40,0.4521)(50,0.3711)};
    \addplot coordinates{(10,0.6079)(20,0.5547)(30,0.5031)(40,0.4210)(50,0.3210)};
    \addplot coordinates{(10,0.2780)(20,0.2200)(30,0.2014)(40,0.1852)(50,0.1324)};
\end{axis}
\end{tikzpicture}
\begin{tikzpicture}
\begin{axis}[
xmax=20,xmin=1,ymin= 0,ymax=1,
xlabel=(b) Econ,ylabel={\em Acc},  width = 4cm, height = 3.7cm,
xmin=10,xmax=50,ymin= 0,ymax=1,
xtick={10,20,...,50},ytick={0,0.2,0.4,...,1}]
    \addplot+[color=black] coordinates{(10,0.9936)(20,0.9744)(30,0.9682)(40,0.9329)(50,0.9079)};
    \addplot+[color=orange]     coordinates{(10,0.8953)(20,0.8523)(30,0.8127)(40,0.7522)(50,0.7125)};
    \addplot+[color=purple] coordinates{(10,0.5509)(20,0.5241)(30,0.5012)(40,0.4221)(50,0.3721)};
    \addplot coordinates{(10,0.6072)(20,0.5621)(30,0.4927)(40,0.4212)(50,0.3751)};
    \addplot coordinates{(10,0.6079)(20,0.5234)(30,0.4735)(40,0.4235)(50,0.3752)};
    \addplot coordinates{(10,0.2978)(20,0.2333)(30,0.1625)(40,0.1221)(50,0.0722)};
\end{axis}
\end{tikzpicture}
\begin{tikzpicture}
\begin{axis}[
xmax=20,xmin=1,ymin= 0,ymax=1,
xlabel=(c) DBLP,ylabel={\em Acc},  width = 4cm, height = 3.7cm,
xmin=10,xmax=50,ymin= 0,ymax=1,
xtick={10,20,...,50},ytick={0,0.2,0.4,...,1}]
    \addplot+[color=black] coordinates{(10,0.9358)(20,0.9126)(30,0.8810)(40,0.8475)(50,0.7996)};
    \addplot+[color=orange]     coordinates{(10,0.9126)(20,0.8897)(30,0.8452)(40,0.7842)(50,0.6948)};
    \addplot+[color=purple] coordinates{(10,0.6458)(20,0.6121)(30,0.5844)(40,0.4655)(50,0.3983)};
    \addplot coordinates{(10,0.7646)(20,0.7316)(30,0.6811)(40,0.6234)(50,0.5977)};
    \addplot coordinates{(10,0.3436)(20,0.2985)(30,0.2230)(40,0.1985)(50,0.1652)};
    \addplot coordinates{(10,0.1713)(20,0.1324)(30,0.0724)(40,0.0685)(50,0.0374)};
\end{axis}
\end{tikzpicture}
\begin{tikzpicture}
\begin{axis}[
xmax=20,xmin=1,ymin= 0,ymax=1,
xlabel=(d) Foursquare,ylabel={\em Acc},  width = 4cm, height = 3.7cm,
xmin=10,xmax=50,ymin= 0,ymax=1,
xtick={10,20,...,50},ytick={0,0.2,0.4,...,1}]
    \addplot+[color=black] coordinates{(10,0.8998)(20,0.8726)(30,0.8010)(40,0.7575)(50,0.6496)};
    \addplot+[color=orange] coordinates{(10,0.021)(20,0.017)(30,0.015)(40,0.01)(50,0.0021)};
    \addplot+[color=purple] coordinates{(10,0.8128)(20,0.7528)(30,0.6534)(40,0.5847)(50,0.5272)};
    \addplot coordinates{(10,0.5214)(20,0.4916)(30,0.4311)(40,0.3734)(50,0.2577)};
    \addplot coordinates{(10,0.3872)(20,0.2585)(30,0.2230)(40,0.1985)(50,0.1252)};
    \addplot coordinates{(10,0.1074)(20,0.0824)(30,0.0624)(40,0.0485)(50,0.0274)};
\end{axis}
\end{tikzpicture}
\\
\ref{named}
\caption{Alignment accuracy according to different levels of the structural noise (\%) on four synthetic datasets.}
\label{Q4plot_structure}
\end{figure*}

\begin{figure*}[t!]
\pgfplotsset{footnotesize,samples=10}
\centering
\begin{tikzpicture}
\begin{axis}[
legend columns=4, 
legend entries={\textsf{Grad-Align},GAlign,CENALP,FINAL},
legend to name=named,
xlabel=(a) Econ,ylabel={\em Acc},  width = 4cm, height = 3.7cm,
xmin=10,xmax=50,ymin= 0.2,ymax=1,
xtick={10,20,...,50},ytick={0.2,0.4,...,1}]
    \addplot+[color=black] coordinates{(10,0.9841)(20,0.9744)(30,0.9654)(40,0.9539)(50,0.9327)};
    \addplot+[color=orange] coordinates{(10,0.8553)(20,0.8324)(30,0.7958)(40,0.7514)(50,0.6895)};
    \addplot+[color=purple] coordinates{(10,0.5509)(20,0.5241)(30,0.4912)(40,0.4621)(50,0.4321)};
    \addplot coordinates{(10,0.7646)(20,0.7042)(30,0.6247)(40,0.5352)(50,0.4851)};

\end{axis}
\end{tikzpicture}
\begin{tikzpicture}
\begin{axis}[
xmax=20,xmin=1,ymin= 0,ymax=1,
xlabel=(b) DBLP,ylabel={\em Acc},  width = 4cm, height = 3.7cm,
xmin=10,xmax=50,ymin= 0.2,ymax=1,
xtick={10,20,...,50},ytick={0.2,0.4,...,1}]
    \addplot+[color=black] coordinates{(10,0.9249)(20,0.9154)(30,0.8721)(40,0.8475)(50,0.8152)};
    \addplot+[color=orange]     coordinates{(10,0.9126)(20,0.8845)(30,0.8215)(40,0.7543)(50,0.7138)};
    \addplot+[color=purple] coordinates{(10,0.6458)(20,0.6125)(30,0.5885)(40,0.5248)(50,0.4857)};
    \addplot coordinates{(10,0.7646)(20,0.7042)(30,0.6247)(40,0.5354)(50,0.4821)};
\end{axis}
\end{tikzpicture}

\ref{named}
\caption{Alignment accuracy according to different levels of the attribute noise (\%) on two synthetic datasets.}
\label{Q4plot_att}
\end{figure*}
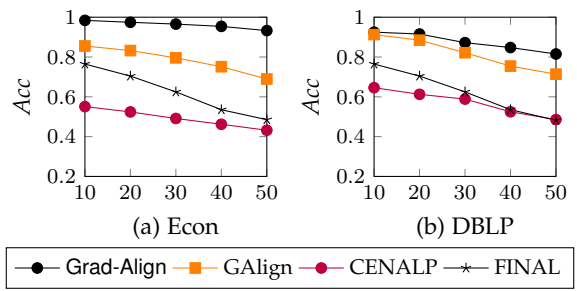
\subsubsection{Comparison with State-of-the-Art Approaches (Q3)}

The performance comparison between our \textsf{Grad-Align} method and five state-of-the-art NA methods, including GAlign~\cite{trung2020adaptive}, CENALP~\cite{du2019joint}, FINAL~\cite{zhang2016final}, PALE~\cite{man2016predict}, and DeepLink~\cite{zhou2018deeplink}, is comprehensively presented in Table \ref{Q3table} with respect to {\em Acc} and $Precision@q$ using three real-world and four synthetic datasets. We note that the hyperparameters  in all the aforementioned state-of-the-art methods are tuned differently according to each individual dataset so as to provide the best performance. In Table \ref{Q3table}, the value with an underline indicates the best performer for each case. We would like to make the following insightful observations:
\begin{itemize}
    \item Our \textsf{Grad-Align} method consistently and significantly outperforms all the state-of-the-art methods regardless of the datasets and the performance metrics.
    \item The second-best performer depends on the datasets, which implies that one does not dominate others among the five state-of-the-art NA methods.

    \item The performance gap between our \textsf{Grad-Align} method ($X$) and the second-best performer ($Y$) is the largest when the Douban dataset is used; the maximum improvement rate of $75.54\%$ is achieved in terms of {\em Acc}, where the improvement rate (\%) is given by $\frac{X-Y}{Y}\times 100$.
    
    \item The Douban dataset is relatively challenging since it exhibits not only quite different scales between two networks $G_s$ and $G_t$ but also the inconsistent network topological structure between $G_s$ and $G_t$ \cite{du2019joint,zhang2016final}. In the dataset, FINAL and GAlign, corresponding to the second and third best performers, respectively, perform satisfactorily by taking advantage of node attribute information; however, other NA methods such as CENALP, PALE, and DeepLink that depend heavily on the topological structure do not perform well.
    
    \item For the non-attributed networks such as Fb-Tw, Facebook (and its noisy version), and Foursquare (and its noisy version), \textsf{Grad-Align} is far superior to state-of-the-art methods. In such datasets, GAlign built upon the GCN model performs the worst. This implies that designing a NA method based solely on GNN models would not guarantee satisfactory performance.
    \item Let us discuss the performance regarding four synthetic datasets. As addressed before, GAlign is quite inferior to other methods on Facebook (not having node attributes), but becomes the second-best performer on Econ and DBLP (having node attributes). In contrast, CENALP and PALE reveal convincing results on Facebook but weakly perform on Econ and DBLP. This implies that state-of-the-art methods are highly dependent on either the topological consistency or the attribute consistency. However, our method is shown to be robust to both structural and attribute noises by virtue of the dual-perception similarity measures.
\end{itemize}

\subsubsection{Robustness to Network Noises (Q4)}

We now compare our \textsf{Grad-Align} method to the five state-of-the-art NA methods in more difficult settings that often occur in real environments: 1) the case in which a large portion of edges in two given networks $G_s$ and $G_t$ are  removed and 2) the case in which a large portion of node attributes in $G_s$ and $G_t$ are missing and replaced with zeros. The performance on {\em Acc} is presented according to different structural and attribute noise levels in Figs. \ref{Q4plot_structure} and \ref{Q4plot_att}, respectively, using synthetic datasets.
\begin{itemize}
    \item In Fig. \ref{Q4plot_structure}, we show how {\em Acc} behaves when we randomly remove $\{10,\cdots,50\}\%$ of existing edges in each dataset. Our findings demonstrate that, while the performance tends to degrade with an increasing portion of missing edges for all the methods, our \textsf{Grad-Align} consistently achieves superior performance compared to state-of-the-art methods. It is also seen that \textsf{Grad-Align} tends to be quite robust to such structural noises especially for datasets having node attributes (i.e., Econ and DBLP). Moreover, since PALE and DeepLink depend solely on the structural consistency of nodes, the performance degradation is significant.
    \item In Fig. \ref{Q4plot_att}, we show how the performance varies when we randomly replace $\{10,\cdots,50\}\%$ of node attributes in each dataset. Due to the fact that node attributes are unavailable on Facebook and Foursquare, the performance is presented using the Econ and DBLP datasets. It is verified that \textsf{Grad-Align} tends to be quite robust to node attribute noises.
\end{itemize}
It is worth noting that \textsf{Grad-Align} reveals strong robustness to the network noises especially for attributed networks. By virtue of the dual-perception similarity, our method is capable of bringing the synergy effect through reinforcing consistencies of nodes. In other words, when the structural noise increases, the embedding similarity plays an important role in guaranteeing high performance, while the Tversky similarity contributes more to the reliable performance  as the attribute noise increases.

\subsubsection{Impacts of Components in \textsf{Grad-Align} (Ablation Studies) (Q5)}
\label{sec5.5.5}
\begin{table}[t!]
\footnotesize 
\centering
\scalebox{1}{
\begin{tabular}{ccccc}
\toprule 
\multicolumn{1}{c}{Method} & \multicolumn{1}{c}{Metric} &\multicolumn{1}{c}{\begin{tabular}[c]{@{}c@{}}Fb-Tw\end{tabular}} & \multicolumn{1}{c}{\begin{tabular}[c]{@{}c@{}}Douban\end{tabular}} & \multicolumn{1}{c}{\begin{tabular}[c]{@{}c@{}}Am-ID\end{tabular}} \\
\midrule
\multirow{4}{*}{\rotatebox{0}{\textsf{Grad-Align}}}
& {\em Acc} & 0.96 &  0.59 &  0.96\\
&$Precision@1$&  0.96&  0.67 &  0.96\\
&$Precision@5$&  0.98 &  0.86 &  \underline{0.97}\\
&$Precision@10$& \underline{0.99} &  0.90 &  \underline{0.98}\\

\multirow{4}{*}{\rotatebox{0}{\textsf{Grad-Align}-1}}
& {\em Acc} & {0.47} &  {0.32} &  {0.85}\\
&$Precision@1$&  {0.47} &  {0.42} &  {0.87}\\
&$Precision@5$&  {0.65} &  {0.55} &  {0.91}\\
&$Precision@10$&  {0.77} &  {0.71} &  {0.94}\\

\multirow{4}{*}{\rotatebox{0}{\textsf{Grad-Align}-2}}
& {\em Acc} &  {0.95} &  {0.28} &  {0.85}\\
&$Precision@1$&  {0.95} &  {0.38} &  {0.86}\\
&$Precision@5$&  {0.97} &  {0.45} &  {0.90}\\
&$Precision@10$&  {0.98} &  {0.55} &  {0.95}\\

\multirow{4}{*}{\rotatebox{0}{\textsf{Grad-Align}-3}}
& {\em Acc} &  {0.05} &  {0.31} &  {0.31}\\
&$Precision@1$&  {0.07} &  {0.41} &  {0.37}\\
&$Precision@5$&  {0.11} &  {0.57} &  {0.49}\\
&$Precision@10$&  {0.17} &  {0.68} &  {0.58}\\

\multirow{4}{*}{\rotatebox{0}{\textsf{Grad-Align-EA}}}
& {\em Acc} & \underline{0.97} &  \underline{0.60} &  \underline{0.97}\\
&$Precision@1$&  \underline{0.97} &  \underline{0.68} &  \underline{0.97}\\
&$Precision@5$&  \underline{0.99} &  \underline{0.87} &  \underline{0.97}\\
&$Precision@10$& \underline{0.99} &  \underline{0.91} &  \underline{0.98}\\
\bottomrule
\end{tabular}}
\caption{Performance comparison among \textsf{Grad-Align} and its variants in terms of the {\em Acc} and {\em Precision@q} on the three real-world datasets. Here, the best case is highlighted using underlines.}
\label{ablation}
\end{table}
In order to discover what role each component plays in the success of the proposed \textsf{Grad-Align} method, we conduct an ablation study by removing each module in our method. Additionally, we empirically show the gain of the edge augmentation module in \textsf{Grad-Align}. 
\begin{itemize}
    \item \textsf{Grad-Align}: This corresponds to the original \textsf{Grad-Align} method without removing any components.
    \item \textsf{Grad-Align-1}: The module of gradual alignment is removed. In other words, all node pairs are found at once using the dual-perception similarity ${\bf S}^{(1)}$ (i.e., $iter=1$).
    \item \textsf{Grad-Align-2}: The module of embedding similarity calculation is removed. That is, \textsf{Grad-Align} is performed only using ${\bf S}_{Tve}^{(i)}$.
    \item \textsf{Grad-Align-3}: The module of Tversky similarity calculation is removed. That is, \textsf{Grad-Align} is performed only using ${\bf S}_{emb}$.
    \item \textsf{Grad-Align-EA}: The edge augmentation module is added to our original \textsf{Grad-Align} method. The threshold $\tau$ is set to 0.7 for the three real-world datasets.
\end{itemize}
The performance comparison among the original \textsf{Grad-Align} and its variants, including three methods with each component removal and \textsf{Grad-Align-EA}, is presented in Table \ref{ablation} with respect to {\em Acc} and {\em Precision@q} using three real-world datasets. In comparison with three methods with each component removal (i.e., ablation studies), one can see that the original \textsf{Grad-Align} method always exhibits potential gains over other variants, which demonstrate that each module plays a critical role together in discovering node correspondences. More interestingly, we observe that the performance gap between \textsf{Grad-Align} and \textsf{Grad-Align-3} tends to be much higher than \textsf{Grad-Align} and other variants especially for the Fb-Tw dataset in which node attributes are not available. This finding indicates that the Tversky similarity calculation is indeed very crucial while the other modules further boost the performance as a supplementary role. Furthermore, in comparison with \textsf{Grad-Align-EA}, it is observed that the performance is enhanced compared to the original \textsf{Grad-Align} method since the edge augmentation module reinforces the structural consistency of given networks over iterations. However, such a gain is possible at the cost of extra computational complexities.

\subsubsection{Computational Complexity (Q6)}
\begin{figure}[t!]
    \centering
    \pgfplotsset{compat=1.11,
    /pgfplots/ybar legend/.style={
    /pgfplots/legend image code/.code={%
       \draw[##1,/tikz/.cd,yshift=-0.25em]
        (0cm,0cm) rectangle (3pt,0.8em);},
   },
}

\begin{tikzpicture}
    \begin{semilogyaxis}[
        width  = 0.85\columnwidth,
        height = 3.3cm,
        major x tick style = transparent,
        ymode=log,
        ybar=0,
        bar width=0.025*\columnwidth,
        ymajorgrids = true,
        ylabel = Execution time (s),
        xlabel = (a) Three real-world datasets,
        symbolic x coords={Fb-Tw, Douban, Am-ID},
        xtick = data,
        enlarge x limits=0.3,
        ymin=0,
        legend cell align=center,
        legend style={at={(0.5,1.1)}, anchor=south,legend columns=3,font=\footnotesize}
    ]

        \addplot[style={black,fill=teal,mark=none}]
            coordinates {(Fb-Tw, 20) (Douban, 65) (Am-ID, 880)};
        \addplot[style={black,fill=lightgray,mark=none}]
            coordinates {(Fb-Tw, 20) (Douban, 103) (Am-ID, 455)};
        \addplot[style={black,fill=brown,mark=none}]
            coordinates {(Fb-Tw, 4512) (Douban,10157) (Am-ID, 61801)};
        \addplot[style={black,fill=pink,mark=none}]
            coordinates {(Fb-Tw, 131) (Douban, 135) (Am-ID, 811)};
        \addplot[style={black,fill=olive,mark=none}]
            coordinates {(Fb-Tw, 151) (Douban, 218) (Am-ID, 1135)};
        \addplot[style={black,fill=orange,mark=none}]
            coordinates {(Fb-Tw, 433) (Douban, 485) (Am-ID, 2945)};

        \legend{\textsf{Grad-Align},GAlign,CENALP,FINAL,PALE,DeepLink} 
    \end{semilogyaxis}
\end{tikzpicture}

    \begin{tikzpicture}
    \begin{semilogyaxis}[
        width  = 0.95\columnwidth,
        height = 3.3cm,
        major x tick style = transparent,
        ymode=log,
        ybar=0,
        xlabel = (b) Four synthetic datasets,
        bar width=0.02*\columnwidth,
        ymajorgrids = true,
        ylabel = Execution time (s),
        symbolic x coords={Facebook, Econ, DBLP, Foursquare},
        xtick = data,
        enlarge x limits=0.3,
        ymin=0,
        legend cell align=center,
        legend style={at={(0.5,1.1)}, anchor=south,legend columns=3,font=\footnotesize}
    ]

        \addplot[style={black,fill=teal,mark=none}]
            coordinates {(Facebook, 23) (Econ, 26) (DBLP, 58) (Foursquare, 2017)};
        \addplot[style={black,fill=lightgray,mark=none}]
            coordinates {(Facebook, 23) (Econ, 39) (DBLP, 112) (Foursquare, 1072)};
        \addplot[style={black,fill=brown,mark=none}]
            coordinates {(Facebook, 4557) (Econ, 5125) (DBLP, 13454) (Foursquare, 184125)};
        \addplot[style={black,fill=pink,mark=none}]
            coordinates {(Facebook, 101) (Econ, 132) (DBLP, 314) (Foursquare, 2701)};
        \addplot[style={black,fill=olive,mark=none}]
            coordinates {(Facebook, 171) (Econ, 181) (DBLP, 545) (Foursquare, 7854)};
        \addplot[style={black,fill=orange,mark=none}]
            coordinates {(Facebook, 458) (Econ, 512) (DBLP, 912) (Foursquare, 10042)};

    
    \end{semilogyaxis}
\end{tikzpicture}
    \caption{The runtime complexity of \textsf{Grad-Align} and five state-of-the-art methods.}
    \label{Q6plot}
\end{figure}
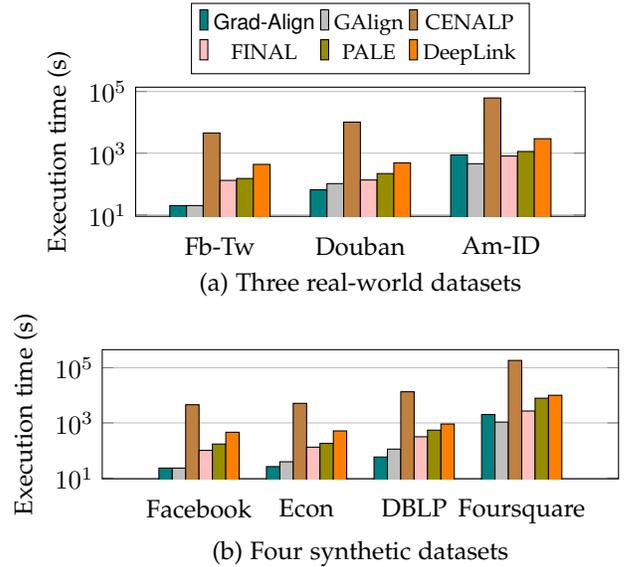

To empirically validate the average runtime complexity of our \textsf{Grad-Align} method, we conduct experiments using the three real-world datasets as well as the four synthetic datasets whose number of nodes is sufficiently large (see Table \ref{datasettable}). Fig. \ref{Q6plot} illustrates the execution time (in seconds) of \textsf{Grad-Align} and five state-of-the-art methods on the three real-world datasets and the four synthetic datasets. It is seen that CENALP has the highest runtime for all the datasets. On the other hand, the computational complexity of \textsf{Grad-Align} is competitive to that of light-weight models such as GAlign and FINAL. Besides, from Fig. \ref{Q2plot}c, due to the fact that reliable performance is still guaranteed even with small $iter$ for the Am-ID dataset having the largest graph size, we can greatly reduce the runtime of \textsf{Grad-Align} by setting $iter$ sufficiently small (e.g., $iter=5$) while maintaining the satisfactory performance.

\section{Concluding Remarks}\label{section 6}
In this paper, we explored an open yet important problem of how much gradual alignment in the NA task is beneficial over the discovery of all node pairs at once. To tackle this challenge, we introduced \textsf{Grad-Align}, a novel NA method that gradually aligns only a part of node pairs iteratively until all node pairs are found across two different networks by making full use of the information of already aligned node pairs having strong consistency in discovering weakly consistent node pairs. To realize our method, we proposed the dual-perception similarity consisting of the embedding similarity and the Tversky similarity. Specifically, we presented an approach to 1) calculating the similarity of multi-layer embeddings based on GNNs using the weight-sharing technique and the layer-wise reconstruction loss, 2) calculating the Tversky similarity to the network imbalance problem that often occurs in practice, and 3) iteratively updating our dual-perception similarity for gradual matching. Additionally, to boost the performance of the original \textsf{Grad-Align}, we developed \textsf{Grad-Align-EA} integrating an edge augmentation module into \textsf{Grad-Align}, which helps reinforce the structural consistency across different networks. Using various real-world and synthetic datasets, we demonstrated that our \textsf{Grad-Align} method remarkably outperforms all state-of-the-art NA methods in terms of the {\em Acc} and {\em Precision@q} while showing significant gains over the second-best performer by up to 75.54\%. We also empirically validated the robustness of \textsf{Grad-Align} to both structural and attribute noises by virtue of our judiciously devised dual-perception similarity. Furthermore, not only the effect of key hyperparameters including $\alpha$ but also the impacts of each module in \textsf{Grad-Align} were comprehensively investigated.

Potential avenues of our future research include the design of an effective GNN model aimed at performing the NA task even in networks without node attributes.

\ifCLASSOPTIONcompsoc
  \section*{Acknowledgments}
\else
  \section*{Acknowledgment}
\fi

This work was supported by the National Research Foundation of Korea (NRF) grant funded by the Korea government (MSIT) under Grants 2021R1A2C3004345 and RS-2023-00220762. An earlier version of this article was presented in part at the AAAI Conference on Artificial Intelligence, Virtual Event, February/March 2022 \cite{park2022grad}.


\ifCLASSOPTIONcaptionsoff
  \newpage
\fi



\bibliographystyle{IEEEtran}
\bibliography{IEEEabrv,1.Citation_list}
%

%

\end{document}